\documentclass[num-refs]{wiley-article}

\usepackage{siunitx}
\usepackage{amssymb}
\usepackage{amsmath}
\usepackage{bm}
\usepackage{caption}
\usepackage{subcaption}
\usepackage{graphicx}
\usepackage{dcolumn}
\usepackage{bm}
\usepackage{epstopdf}
\usepackage{epsfig}
\usepackage{hyperref}
\usepackage{url}

\newcommand{\R}{\mathbb{R}}

\newcommand{\Z}{\mathbb{Z}}

\newcommand\inp[2]{\left\langle #1, #2 \right\rangle}

\usepackage{anyfontsize}

\papertype{Original Article}
\paperfield{Journal Section}

\title{Maxwell fronts in the discrete nonlinear Schr\"odinger equations with competing nonlinearities}

\author{Farrell\ Theodore\ Adriano}
\author{Hadi Susanto}

\affil[]{Department of Mathematics, Khalifa University, PO Box 127788, Abu Dhabi, United Arab Emirates}

\corraddress{Farrell\ Theodore\ Adriano, Department of Mathematics, Khalifa University, PO Box 127788, Abu Dhabi, United Arab Emirates}
\corremail{100065157@ku.ac.ae}

\fundinginfo{Khalifa University, Competitive Internal Research Awards Grant (No.\ 8474000413/CIRA-2021-065), and Research \& Innovation Grant (No.\ 8474000617/RIG-2023-031 and No.\ 8474000789/RIG-S-2024-070)}

\runningauthor{Adriano \& Susanto}

\begin{document}

\begin{frontmatter}
\maketitle

\begin{abstract}
In discrete nonlinear systems, the study of nonlinear waves has revealed intriguing phenomena in various fields such as nonlinear optics, biophysics, and condensed matter physics. Discrete nonlinear Schr\"odinger (DNLS) equations are often employed to model these dynamics, particularly in the context of Bose-Einstein condensates and optical waveguide arrays. While the classical DNLS with cubic nonlinearity admits well-known solitonic solutions, the introduction of competing nonlinearities, such as quadratic-cubic and cubic-quintic terms, gives rise to new behaviors, including multistability and front formation. One such emergent structure, the Maxwell front, is characterized by stationary interfaces between two energetically equivalent steady states, occurring at a critical parameter known as the Maxwell point.

This paper investigates the existence and stability of Maxwell fronts in DNLS models with competing nonlinearities. Specifically, we examine the quadratic-cubic nonlinearity, as found in the discrete quantum droplets equation, and the cubic-quintic nonlinearity, both of which exhibit multistability. We explore the persistence of Maxwell fronts in both the anticontinuum limit (where the coupling between lattice sites is weak) and the continuum limit (where the coupling is strong). The stability of these fronts is analyzed through linear stability analysis, utilizing eigenvalue counting arguments and exponential asymptotic techniques. Our results provide new insights into multistability, front dynamics, and the role of competing nonlinearities in discrete wave systems.

The main contributions of this work include the characterization of Maxwell fronts in DNLS equations with competing nonlinearities, the analysis of their stability across different coupling regimes, and the application of novel asymptotic methods to investigate their behavior in the continuum limit.

\keywords{Discrete Nonlinear Schr\"odinger Equation, Maxwell Fronts, Competing Nonlinearities, Multistability, Stability Analysis, Exponential Asymptotics}
\end{abstract}
\end{frontmatter}

\section{Introduction}\label{sec1}

Nonlinear wave phenomena in discrete systems play a fundamental role in various areas of physics, including nonlinear optics, biophysics, and condensed matter~\cite{kivshar2003optical,braun2004frenkel,kevrekidis2009discrete,kevrekidis2015defocusing}. These dynamics are often modeled by discrete nonlinear Schr\"odinger (DNLS) equations, particularly in the study of Bose-Einstein condensates (BECs)~\cite{Bloch2005,Pethick_Smith_2008} and optical waveguide arrays~\cite{PhysRevLett.81.3383,PhysRevLett.86.3296}. In these systems, the DNLS typically incorporates a Kerr-type cubic nonlinearity. Depending on the sign of the nonlinear coefficient, the DNLS can support bright solitons in the focusing case and dark solitons in the defocusing case. The stability of these localized states has been widely studied~\cite{Pelinovsky2005,pelinovsky2008stability,PhysRevE.75.066608,Kapitula2001,PhysRevLett.82.85}.

While the classical DNLS with cubic nonlinearity admits a well-understood family of solitonic solutions, the introduction of competing nonlinearities—such as quadratic-cubic or cubic-quintic terms—gives rise to new behaviors, including multistability and front formation. This paper investigates one such emergent structure: the Maxwell front. A Maxwell front is a stationary interface between two energetically equivalent steady states, occurring at a critical parameter known as the Maxwell point.

Variants of the DNLS equation introduce additional phenomena absent in the purely cubic case. One such feature is multistability, as observed in the discrete quantum droplets equation~\cite{kusdiantara2024analysis,zhao2021discrete}, which was proposed in~\cite{petrov2015quantum} to account for the Lee-Huang-Yang (LHY) effect~\cite{lee1957eigenvalues}. The LHY correction arises from quantum fluctuations beyond the mean-field approximation, introducing a local quartic self-repulsive term in the nonlinear Schr\"odinger equation. This term counteracts the mean-field cubic attraction, thereby stabilizing localized structures that would otherwise collapse. The interplay of nonlinearities introduces a quadratic term in the governing equation, leading to a richer energy landscape. Studies have shown that this competition results in multistability and stability switching in optical lattices~\cite{zhou2019dynamics,dong2020multi,kartashov2024enhanced}, including their discrete counterparts~\cite{zhao2021discrete,kusdiantara2024analysis}.

In such systems, multistability manifests through the presence of multiple coexisting temporally stable uniform steady states. This is typically a result of the pinning phenomenon in discrete or spatially extended dynamical systems, particularly those with bistable configurations~\cite{susanto2011variational,matthews2011variational}. The existence of bistable uniform states permits the formation of front solutions connecting these states. 
At a specific parameter value---the so-called \emph{Maxwell point}---the stable uniform states are energetically equivalent. As a result, 
stationary front solutions, referred to as \emph{Maxwell fronts}, may arise. This balance in energy landscape leads to the phenomenon of front pinning and homoclinic snaking over a parameter window, as discussed in~\cite{kusdiantara2024analysis} for the discrete case and in~\cite{Kusdiantara2025} for the continuum case. The multistability of the discrete quantum droplets equation can be analyzed through the system's effective potential. Due to the presence of competing nonlinear terms, the relative energies of steady states vary with the bifurcation parameter. At the Maxwell point, the energies of the coexisting steady states coincide, giving rise to these stationary fronts. 

These stationary fronts are similar to the defocusing dark solitons of the cubic DNLS in the sense that they are heteroclinic connections between two uniform steady states. However, they differ from each other in the following way: dark solitons asymptote to nonzero backgrounds of equal magnitude, whereas the Maxwell front asymptote to backgrounds of different magnitude at either end of the far-field. Moreover, DNLS dark solitons exist for a large range of parameter values \cite{PhysRevE.75.066608}, in contrast to the Maxwell front which exists only at the Maxwell point. These differences alter the existence and stability properties of these fronts.

In this work, we investigate the existence and stability of Maxwell fronts in DNLS models with competing nonlinearities, focusing on the quadratic-cubic nonlinearity (as in the discrete quantum droplets equation) and the cubic-quintic nonlinearity, which also exhibits multistability~\cite{dean2015orientation}. We analyze these fronts in both the anticontinuum limit, where the coupling between lattice sites is weak, and the continuum limit, where the coupling is strong. Of particular interest are two canonical configurations: \emph{onsite Maxwell fronts}, centered at lattice sites, and \emph{intersite Maxwell fronts}, centered between lattice sites. These configurations persist across the coupling parameter range from the anticontinuum to the continuum limit.

The stability of these front solutions is examined using linear stability analysis. This includes an asymptotic analysis of the associated linearized operator and an eigenvalue counting argument, following the approach in~\cite{Chugunova2010,pelinovsky2008stability}. To facilitate the analysis in the continuum regime, we also employ the method of exponential asymptotics~\cite{KING_CHAPMAN_2001,dean2015orientation} to derive precise asymptotic expressions for the Maxwell fronts and subsequently utilize them in the stability analysis. A rigorous proof of the orbital stability of spatially continuous fronts has also been presented recently in \cite{holmer2025orbital}.

The main contributions of this paper are as follows: (i) we establish the existence of Maxwell fronts in discrete systems with quadratic-cubic and cubic-quintic nonlinearities, (ii) we analyze their persistence from the anticontinuum to the continuum limit, and (iii) we determine their linear stability through eigenvalue counting arguments and exponential asymptotic techniques. These findings provide new insights into multistability and front dynamics in discrete nonlinear wave systems.

The remainder of this paper is organized as follows. In Section~\ref{sec2}, we introduce the general DNLS model with competing nonlinearities and describe the framework for the linear stability analysis. Section~\ref{sec3} presents the analysis for the quadratic-cubic DNLS in both the anticontinuum and continuum limits. In Section~\ref{sec4}, we briefly outline the analogous analysis for the cubic-quintic nonlinearity and present corresponding results. Finally, Section~\ref{sec5} concludes the paper.

\section{Mathematical model}\label{sec2}
In this study, we consider the DNLS equation of the form
\begin{equation}\label{eqn:general_dnls}
    i \dot{\phi}_{n} = -\dfrac{C}{2}\Delta \phi_{n} - \mu \phi_{n} + F(|\phi_{n}|)\phi_{n},
\end{equation}
where $\phi_{n}$ denotes the complex-valued amplitude, $\mu$ is a parameter, $C$ is the coupling parameter, and $\Delta$ is the discrete Laplacian operator:
\begin{equation} \label{eqn:Delta}
    \Delta \phi_{n} = \phi_{n-1} - 2\phi_{n} + \phi_{n+1},
\end{equation}
and $F(|\phi_{n}|)$ represents the nonlinearity. The nonlinearities considered are the quadratic-cubic nonlinearity:
\begin{equation}\label{eqn:2-3_nonlinearity}
    F(|\phi_{n}|) = -|\phi_{n}| + |\phi_{n}|^{2},
\end{equation}
and the cubic-quintic nonlinearity:
\begin{equation}\label{eqn:3-5_nonlinearity}
    F(|\phi_{n}|) = -|\phi_{n}|^{2} + |\phi_{n}|^{4}.
\end{equation}
The opposing signs of the terms in $F$ represent the competing effects of the nonlinearity. As will be shown for both types of nonlinearities, this property leads to multistability in the system, i.e., the existence of multiple steady uniform states for a specific parameter value, $\mu = \mu_{M}$, which is referred to as the Maxwell point. At the Maxwell point, the two stable steady states have equal energy, so no state is energetically preferred over the other. Consequently, this allows for the existence of stationary front solutions, which we call the Maxwell fronts. The stationary Maxwell fronts satisfy the time-independent equation
\begin{equation} \label{general_dnls_time_indep}
    0 = -\dfrac{C}{2}\Delta \phi_{n} - \mu \phi_{n} + F(|\phi_{n}|)\phi_{n},
\end{equation}
where $\phi_{n}$ asymptotes to each of the stable uniform states as $n \to \pm \infty$. We shall restrict our attention to real-valued fronts, such that $\phi_{n} \in \R$ for all $n \in \Z$. The study of Maxwell fronts and their stability will be the main focus of this work.

The stability of these fronts will be analyzed through linear stability analysis. Denoting a stationary solution to \eqref{eqn:general_dnls} as $\phi_{n}(t) = \varphi_{n}$, we linearize around the stationary solution by introducing a perturbation of the form $\phi_{n} = \varphi_{n} + (p_{n} + i q_{n})$, where $p_{n} = p_{n}(t)$ and $q_{n} = q_{n}(t)$ are small perturbations, i.e., $|p_{n}|,|q_{n}| \ll 1$. Writing the perturbations as $p_{n} = u_{n} e^{\lambda t}$ and $q_{n} = v_{n} e^{\lambda t}$ leads to the eigenvalue problem
\begin{equation} \label{eqn:eigval_eqn}
    \lambda \begin{bmatrix}
        u \\ v
    \end{bmatrix} = \begin{bmatrix}
        O & L_{-} \\ -L_{+} & O
    \end{bmatrix}\begin{bmatrix}
        u \\ v
    \end{bmatrix} = \mathcal{L} \begin{bmatrix}
        u \\ v
    \end{bmatrix},
\end{equation}
where $u,v \in \ell^{2}(\Z)$. This system will determine the stability of the stationary solution $\varphi$. The corresponding linearized operators $L_{\pm}$ 
are given by
\begin{align}
    L_{-} &= -\dfrac{C}{2}\Delta - \mu + F(|\varphi_{n}|)
    , \label{eqn:Lm} \\
    L_{+} &= -\dfrac{C}{2}\Delta -\mu + F'(|\varphi_{n}|)|\varphi_{n}| + F(|\varphi_{n}|). \label{eqn:Lp} 
\end{align}

Due to the Hamiltonian nature of the system, the eigenvalues $\lambda$ of $\mathcal{L}$ appear in quartets: if $\lambda$ is an eigenvalue of $\mathcal{L}$, then $-\lambda$, $\overline{\lambda}$, and $-\overline{\lambda}$ (where the bar denotes complex conjugation) are also eigenvalues. Consequently, the stability of the solution is determined by the real part of the eigenvalues of $\mathcal{L}$. Specifically, if there exists an eigenvalue $\lambda$ such that $\text{Re}(\lambda) \neq 0$, the corresponding stationary solution is unstable; otherwise, it is neutrally stable.

\section{Quadratic-Cubic DNLS} \label{sec3}
We consider the quadratic-cubic DNLS, i.e., we take $F$ to be \eqref{eqn:2-3_nonlinearity}, and the system \eqref{eqn:general_dnls} then becomes
\begin{equation}\label{eqn:quad_cubic_dnls}
    i \dot{\phi}_{n} = -\dfrac{C}{2}\Delta \phi_{n} - \mu \phi_{n} - |\phi_{n}|\phi_{n} + |\phi_{n}|^{2}\phi_{n},
\end{equation}
where the stationary equation is
\begin{equation} \label{eqn:quad_cubic_dnls_stationary}
    0 = -\dfrac{C}{2}\Delta\varphi_{n} - \mu \varphi_{n} - |\varphi_{n}|\varphi_{n} + |\varphi_{n}|^{2}\varphi_{n}.
\end{equation}

The system \eqref{eqn:quad_cubic_dnls} is Hamiltonian, with the corresponding energy/Hamiltonian given by
\begin{equation}
    E(\phi) = \sum_{n = -\infty}^{\infty}\left[\dfrac{C}{2}|\phi_{n}-\phi_{n-1}|^{2} + \dfrac{\mu}{2}|\phi_{n}|^{2} + \dfrac{1}{3}|\phi_{n}|^{3} - \dfrac{1}{4}|\phi_{n}|^{4}\right].
\end{equation}
As noted in the previous section, system \eqref{eqn:quad_cubic_dnls} exhibits multistability. This can be seen by analyzing the stationary uniform solutions of \eqref{eqn:quad_cubic_dnls} in the form $\phi_{n}(t) = \varphi e^{i\theta_{n}}$ where $\varphi \geq 0$ and $\theta_{n} \in \R$, these give
\begin{equation}\label{eqn:2-3_uniform_solutions}
    \varphi^{\mathrm{low}} = 0, \quad \varphi^{\mathrm{mid}} = \dfrac{1-\sqrt{4\mu+1}}{2}, \quad \varphi^{\mathrm{up}} = \dfrac{1+\sqrt{4\mu+1}}{2},
\end{equation}
and the phase $\theta_{n} \in \R$ is arbitrary. The lower (trivial) uniform solution $\varphi^{\mathrm{low}}$ exists for all values of $\mu$, the middle solution $\varphi^{\mathrm{mid}}$ exists whenever $\mu \in [-1/4,0]$, and the upper solution $\varphi^{\mathrm{up}}$ exists when $\mu \geq -1/4$.
Upon doing the linear stability analysis as outlined in Section \ref{sec2}, with perturbations of the form $u_{n},v_{n} = e^{ikn}$, we obtain the dispersion relation for the spectrum $\lambda$ of the linearized operator $\mathcal{L}$ at the uniform state $\varphi$ as
\begin{equation}\label{eqn:2-3 disp relation}
    \lambda^{2} = -z^{2} + \varphi F'\left(\left|\varphi\right|\right)z
\end{equation}
where $z = C(\cos k -1) + \mu - F(|\varphi|)$. As $\lambda^{2} \in \R$ for all wavenumbers $k \in \R$, the uniform state $\varphi$ is unstable if there exists a $k$ such that $\lambda^{2} > 0$. 

It is now apparent that the lower (trivial) uniform solution $\varphi^{\mathrm{low}} = 0$ is always stable as the associated spectrum 
$\lambda = \pm i\left|\mu + C(\cos k - 1)\right|$ is purely imaginary for all wavenumbers $k$. For the nontrivial uniform solutions, $\mu - F(|\varphi|) = 0$ and $z$ reduces to $z = C(\cos k -1)$. In the case of the upper solution, we have $\varphi^{\mathrm{up}}F'\left(\left|\varphi^{\mathrm{up}}\right|\right) \geq 0$ for any $\mu \geq -1/4$, thus $\lambda^{2} \leq 0$ for any $k \in \R$, i.e. $\varphi^{\mathrm{up}}$ is always stable. Therefore, we have bistability of $\varphi^{\mathrm{low}}$ and $\varphi^{\mathrm{up}}$ for $\mu \geq -1/4$. On the other hand, the middle uniform solution $\varphi^{\mathrm{mid}}$ is unstable whenever $-1/4 < \mu < 0$ since $\varphi^{\mathrm{mid}}F'(|\varphi^{\mathrm{mid}}|) < 0$, thus $\lambda^{2} > 0$ for some small $k$.

Even though we have bistable uniform solutions $\varphi^{\mathrm{low}}$ and $\varphi^{\mathrm{up}}$ for $\mu > -1/4$, for certain values of $\mu$, specifically when $\mu > -2/9$, the upper uniform solution $\varphi^{\mathrm{up}}$ has a lower potential than the lower uniform solution, i.e. $E(\varphi^{\mathrm{up}}) < E\left(\varphi^{\mathrm{low}}\right)$, making it energetically favorable. When $\mu < -2/9$, the situation is reversed. A critical value of $\mu$, at which $\varphi^{\mathrm{low}}$ and $\varphi^{\mathrm{up}}$ have the same potential, is called the Maxwell point. This occurs when $\mu = \mu_{M} = -2/9$. This allows for stationary front solutions that connect the two stable uniform solutions, $\varphi^{\mathrm{low}}$ and $\varphi^{\mathrm{up}}$. Note that in this case, it suffices to consider real stationary solutions $\varphi_{n}$. This can be seen by noting that from the stationary equation \eqref{eqn:quad_cubic_dnls_stationary}, the quantity $J = \overline{\varphi}_{n}\varphi_{n+1} - \varphi_{n}\overline{\varphi}_{n+1}$ is invariant in $n$. As we assumed that $\varphi_{n} \to \varphi^{\mathrm{low}} = 0$ as $n \to -\infty$ (or $n \to \infty$), then $J = 0$ necessarily. This implies that either $\varphi_{n} = 0$ or $(\varphi_{n+1}/\varphi_{n}) = \overline{(\varphi_{n+1}/\varphi_{n})}$, i.e. $\varphi_{n}$ and $\varphi_{n+1}$ must have the same phase (modulo $\pi$). Therefore, we may assume without loss of generality that $\varphi_{n} \in \R$ for all $n \in \Z$.

At the Maxwell point, the three uniform solutions are 
\begin{equation}
    \varphi^{\mathrm{low}} = 0, \quad \varphi^{\mathrm{mid}} = \frac{1}{3}, \quad \varphi^{\mathrm{up}} = \frac{2}{3},
\end{equation}
and specific stationary front solutions persist for all values of $C$, where, in the continuum limit $(C \to \infty)$, they correspond to the heteroclinic orbits connecting the two fixed points $\varphi^{\mathrm{low}}$ and $\varphi^{\mathrm{up}}$. This can be seen through a phase plane analysis of the continuum analogue of stationary equation \eqref{eqn:quad_cubic_dnls_stationary}, which is given by
\begin{equation}\label{eqn:2-3 continuum eqn}
    0 = -\partial_{x}^{2}\varphi - \mu \varphi - |\varphi|\varphi + |\varphi|^{2}\varphi,
\end{equation}
where $x = n/\sqrt{C/2}$ and $\varphi_{n} \equiv \varphi(x)$ as $C \to \infty$.
At the Maxwell point $(\mu = \mu_{M})$, we have a heteroclinic orbit connecting $\varphi^{\mathrm{low}}$ and $\varphi^{\mathrm{up}}$ given by
\begin{equation} \label{eqn:cont_front_solution}
    \varphi(x) = \frac{2/3}{1+e^{-\sqrt{2}(x-x_{0})/3}},
\end{equation}
where $x_{0} \in \R$ is an arbitrary constant.
This heteroclinic orbit corresponds to the Maxwell front solution in the continuum limit. When $\mu < \mu_{M}$, as the upper uniform solution has lower energy, we have the so-called bubbles: a homoclinic orbit emanating from $\varphi^{\mathrm{up}}$. On the other side of $\mu_{M}$, when $\mu > \mu_{M}$, the lower uniform solution is of lower energy, and we have the so-called droplets: a homoclinic orbit from $\varphi^{\mathrm{low}}$. The phase plane representation of these solutions is illustrated in Figure \ref{fig:2-3_phase_portrait}.
\begin{figure}[tbhp]
    \centering
    \subfloat[$\mu=-0.2244$]{\includegraphics[scale=0.31]{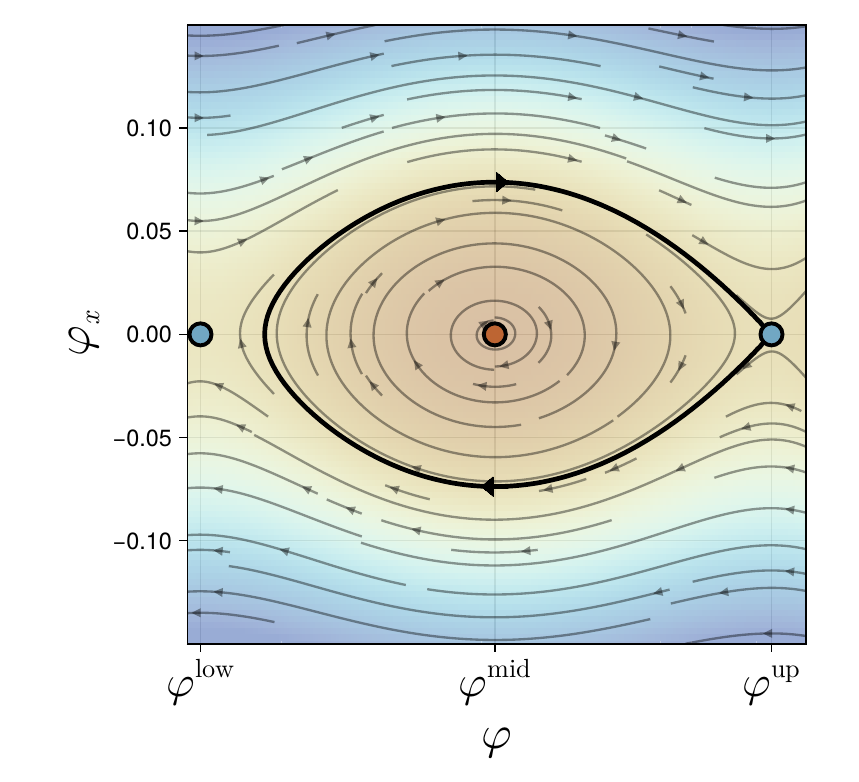}\label{subfig:phase_Mp_left}}
		\subfloat[$\mu={\mu^{(M)}}$]{\includegraphics[scale=0.31]{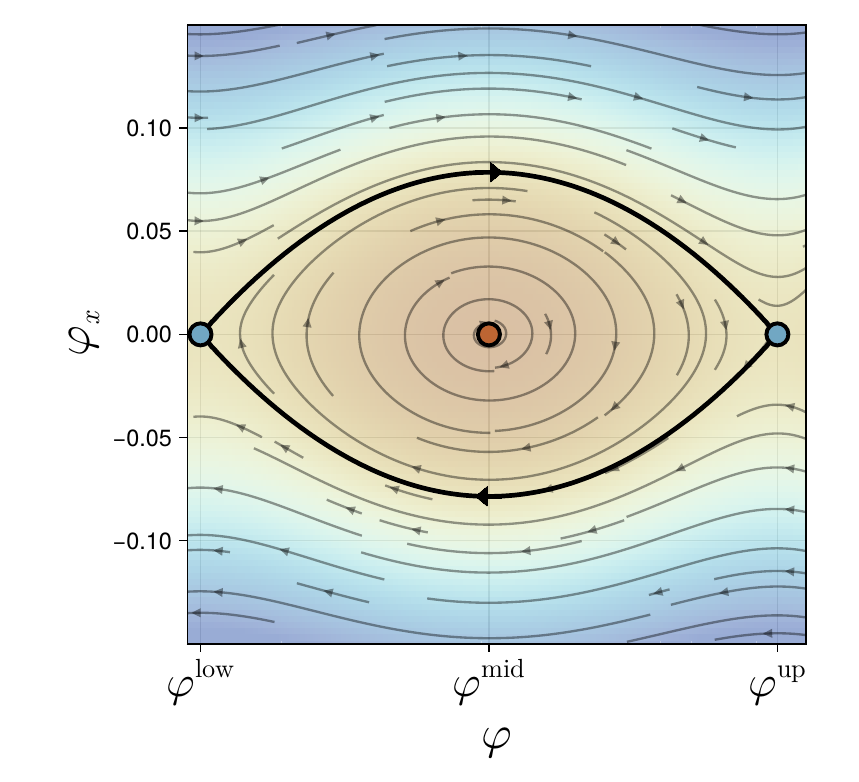}\label{subfig:phase_Mp}}
		\subfloat[$\mu=-0.22$]{\includegraphics[scale=0.31]{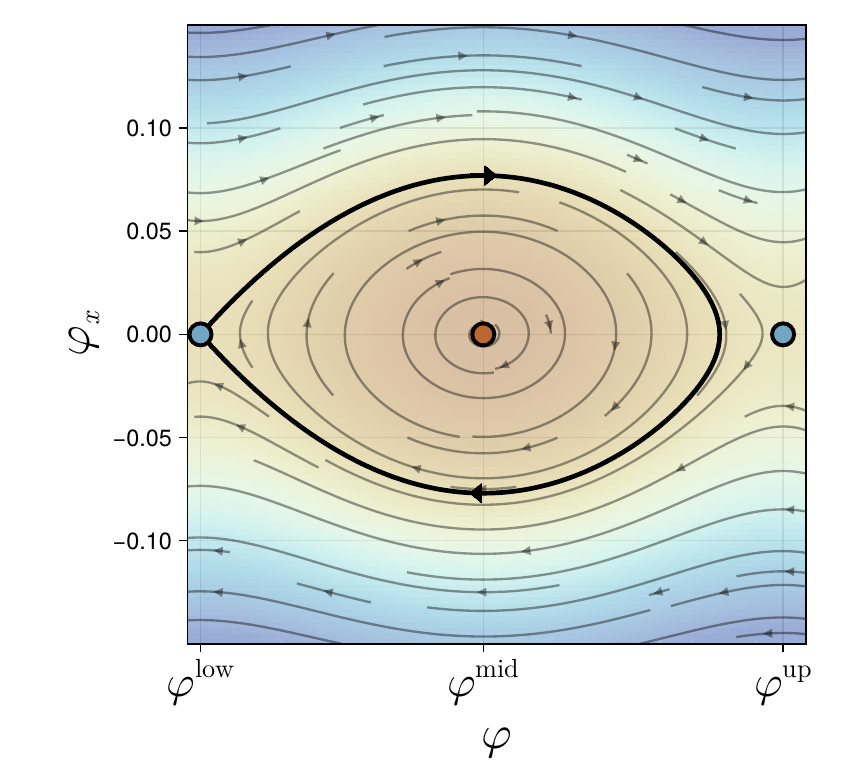}\label{subfig:phase_Mp_right}}
    \caption{Phase portrait of the quadratic-cubic stationary equation \eqref{eqn:quad_cubic_dnls_stationary} in the continuum limit $C \to \infty$ i.e. Eq. \eqref{eqn:2-3 continuum eqn} for values of $\mu$ around the Maxwell point $\mu_{M} = -2/9$. The left panel shows the case of $\mu = -0.2244 < \mu_{M}$, the middle panel shows the case $\mu = \mu_{M}$, and the right panel shows the case $\mu = -0.22 > \mu_{M}$. The black curves in each panel represent the homoclinic/heteroclinic orbits connecting the uniform states.}
    \label{fig:2-3_phase_portrait}
\end{figure}


\subsection{Anticontinuum Limit} \label{sec3.1}
\subsubsection{\texorpdfstring{$C = 0$}{C = 0}}  \label{sec3.1.1}
In the anticontinuum limit ($C = 0$), the stationary equation \eqref{eqn:quad_cubic_dnls_stationary} becomes decoupled, and the stationary solutions are characterized by $\varphi_{n}$ taking one of the three values: $\varphi^{\mathrm{low}}$, $\varphi^{\mathrm{mid}}$, or $\varphi^{\mathrm{up}}$. 
In this limit, fronts take the general form $ \bm{\varphi} = (\dots, \varphi^{\mathrm{low}},\varphi^{\mathrm{low}},\bm{\tilde{\varphi}},\varphi^{\mathrm{up}},\varphi^{\mathrm{up}},\dots)$ where $\bm{\tilde{\varphi}} \in \R^{N}$ is of any length $N \geq 1$ with $\tilde{\varphi}_{1} \in \{\varphi^{\mathrm{mid}},\varphi^{\mathrm{up}}\}$, $\tilde{\varphi}_{N} \in \{\varphi^{\mathrm{low}},
\varphi^{\mathrm{mid}}\}$, and $\tilde{\varphi}_{j}$ taking the value of any one out of the three uniform states $\varphi^{\mathrm{low}},\varphi^{\mathrm{mid}}$, or $\varphi^{\mathrm{up}}$ for $2 \leq j \leq N-1$. The fronts are then characterized by its values central sites, which is given by $\bm{\tilde{\varphi}}$. With this, we can associate to each front configuration a code given by $\bm{\tilde{\varphi}}$. For example, the code $(ulml)$ corresponds to the front with $\bm{\tilde{\varphi}} = (\varphi^{\mathrm{up}},\varphi^{\mathrm{low}},\varphi^{\mathrm{mid}},\varphi^{\mathrm{low}})$. 

In order to investigate the Maxwell fronts at $\mu = \mu_{M}$, we will focus our analysis on two particular limiting configurations of the Maxwell fronts. They are the two simplest fronts of codes with length $N \leq 1$. The first configuration is the onsite Maxwell front, where the center of the front is located at a lattice point. This configuration takes the following form:
\begin{equation}\label{eqn:ac_onsite_maxwell_front}
    \varphi_{n} = 
    \begin{cases}
        \varphi^{\mathrm{low}}, &n \leq -1, \\
        \varphi^{\mathrm{mid}}, &n = 0, \\
        \varphi^{\mathrm{up}}, &n \geq 1.
    \end{cases}
\end{equation}
The onsite front can be regarded as the front with code $(m)$.
The second configuration is the intersite Maxwell front, where the center of the front is located between two lattice points. This configuration is given by:
\begin{equation} \label{eqn:ac_intersite_maxwell_front}
    \varphi_{n} = 
    \begin{cases}
        \varphi^{\mathrm{low}}, &n \leq 0, \\
        \varphi^{\mathrm{up}}, &n \geq 1.
    \end{cases}
\end{equation}
The intersite front can be regarded as the front with $\bm{\tilde{\varphi}}$ of length $N = 0$ and will be labeled the code $(0)$.
As we shall see, these are the two configurations that persist up to the continuum limit.

The question of linear stability can now be addressed by analyzing the corresponding eigenvalue problem, which may be rewritten in the equivalent form
\begin{equation}\label{eqn:eigval_eqn_LpLm}
    L_{+}L_{-} v = -\lambda^{2} v.
\end{equation}
In the anticontinuum limit, $L_{+}$ and $L_{-}$ become multiplicative operators, and the eigenvalue problem decouples. Thus, we can determine the eigenvalues by considering the contribution from sites where $\varphi_{n}$ takes distinct values. For lattice sites $n$ such that $\varphi_{n} = \varphi^{\mathrm{low}} = 0$, we have
\begin{equation}
    (L_{+}L_{-}v)_{n} = \mu^{2}v_{n}.
\end{equation}
Therefore, the zero sites correspond to the pair of imaginary eigenvalues $\lambda = \pm i |\mu|$. For the nonzero sites, i.e., $n$ such that $\varphi_{n} = \varphi^{\mathrm{mid}}$ or $\varphi^{\mathrm{up}}$, we obtain
\begin{equation}
    (L_{+}L_{-}v)_{n} = 0,
\end{equation}
which gives the eigenvalue $\lambda = 0$.
\par
Thus, in the anticontinuum limit, the spectrum of $\mathcal{L}$ for both the onsite and intersite configurations consists of $\lambda = 0$, corresponding to the nonzero sites, and $\lambda = \pm i |\mu|$, corresponding to the zero sites of each configuration.


We also note that in the $C = 0$ limit, fronts can exist even if $\mu \neq \mu_{M}$. The subsequent analysis in the case of small $C$ holds for the fronts whenever the three uniform states \eqref{eqn:2-3_uniform_solutions} exist, i.e. $\mu \in (-1/4,0)$. However these fronts do not persist up until the continuum limit as $C \to \infty$ (see discussion in Section \ref{sec3.3}).

\subsubsection{\texorpdfstring{$0 < C \ll 1$}{C small}} \label{sec3.1.2}
When $0 < C \ll 1$, the profile of the 
stationary fronts in the anticontinuum limit will be perturbed due to the presence of coupling between the sites. However, since the coupling is weak, we can analyze the perturbation by writing the stationary solution as an asymptotic expansion in powers of $C$:
\begin{equation} \label{eqn:ac_asymp_expansion_phi}
    \varphi_{n} = \varphi_{n}^{(0)} + C \varphi_{n}^{(1)} + O(C^{2}) \quad (C \ll 1),
\end{equation}
where $\varphi^{(0)}$ will be taken to be the stationary front
in the anticontinuum limit. Substituting \eqref{eqn:ac_asymp_expansion_phi} into \eqref{eqn:quad_cubic_dnls_stationary}, we find that the $O(C)$ correction term to $\varphi^{(0)}$ is given by
\begin{equation} \label{eqn:ac_phi_O(C)}
    \varphi_{n}^{(1)} = \dfrac{\Delta \varphi_{n}^{(0)}}{2\left(3\left|\varphi_{n}^{(0)}\right|^{2} - 2\left|\varphi_{n}^{(0)}\right| - \mu\right)}.
\end{equation}
Thus, for the onsite Maxwell front, the correction terms are
\begin{equation}
    \varphi_{n}^{(1)} = 
    \begin{cases}
        \frac{3}{4}, &n = -1, \\
        0, &n = 0, \\
        -\frac{3}{4}, &n = 1, \\
        0, & \text{otherwise}.
    \end{cases}
\end{equation}
Whereas for the intersite Maxwell front, they are
\begin{equation}
    \varphi_{n}^{(1)} = 
    \begin{cases}
        \frac{3}{2}, &n = 0, \\
        -\frac{3}{2}, &n = 1, \\
        0, & \text{otherwise}.
    \end{cases}
\end{equation}
To establish the linear stability of the onsite and intersite Maxwell fronts, we will use an eigenvalue counting argument (following \cite{pelinovsky2008stability,Chugunova2010}) to determine the number of eigenvalues of $\mathcal{L}$ and their corresponding types by analyzing the operators $L_{+}$ and $L_{-}$ separately. The main result is the following theorem, which applies generically to the eigenvalue problem of the form \eqref{eqn:eigval_eqn_LpLm} under several assumptions on $L_{\pm}$ and the stationary solution $\bm{\varphi}$.

\begin{theorem} \label{thm:eigval_count}
    Assume that the stationary solution $\bm{\varphi}$ is a front solution satisfying $\varphi_{n} \to \varphi^{\mathrm{low}}$ as $n \to -\infty$ and $\varphi_{n} \to \varphi^{\mathrm{up}}$ as $n \to \infty$, and that $\bm{\varphi}$ does not change signs on $\Z$. Suppose further that $L_{\pm}$ have empty kernels in $\ell^{2}(\Z)$, and let $n(L_{\pm})$ denote the number of negative eigenvalues of $L_{\pm}$. Under these conditions, the point spectrum of $\mathcal{L}$ consists of $n(L_{+})$ real nonzero eigenvalues and no imaginary or complex eigenvalues.
\end{theorem}

\begin{proof}
    Since $L_{+}$ has an empty kernel, we can consider \eqref{eqn:eigval_eqn_LpLm} equivalently as
    \begin{equation} \label{eqn:eigval_eqn_Lm_Lp_inv}
        L_{-}v = -\lambda^{2}L_{+}^{-1}v.
    \end{equation}
    Furthermore, since $L_{-}\bm{\varphi} = 0$, $L_{-}$ is not invertible, so we shift \eqref{eqn:eigval_eqn_Lm_Lp_inv} to 
    \begin{equation}
        \left(L_{-} + \delta L_{+}^{-1}\right)v = (-\lambda^{2} + \delta)L_{+}^{-1}v,
    \end{equation}
    for sufficiently small $\delta > 0$. 
    \par
    Now, by analyzing the far-field as $|n| \to \infty$, since $\varphi_{n} \to \varphi_{\mathrm{up}}$ as $n \to \infty$ and $\varphi_{n} \to \varphi_{\mathrm{low}}$ as $n \to -\infty$, Weyl's Essential Spectrum Lemma gives that the essential spectrum of $L_{-}$ and $L_{+}$ are respectively bounded below by $0$ and $-\mu_{M} = 2/9$. The properties P1 and P2 of \cite{Chugunova2010} are satisfied, and we may invoke Theorem 3 of \cite{Chugunova2010} to obtain a relation between the eigenvalue counts of $L_{-} + \delta L_{+}^{-1}$, $L_{+}$, and $\mathcal{L}$ as
    \begin{align}
        n(L_{+}) &= N_{r}^{-} + N_{i}^{-} + N_{c} \label{eqn:eigval_count_plus}\\
        n(L_{-} + \delta L_{+}^{-1}) &= N_{r}^{+} + N_{i}^{-} + N_{c} \label{eqn:eigval_count_min}
    \end{align}
    where $n(A)$ denotes the number of negative eigenvalues of the operator $A$, $N_{r}^{\pm}$ denotes the number of real eigenvalues of $\mathcal{L}$ such that $\inp{v}{L_{+}^{-1}v} \geq 0$ ($\leq 0$), $N_{i}^{-}$ denotes the number of imaginary eigenvalues (with positive imaginary part) of $\mathcal{L}$ with $\inp{v}{L_{+}^{-1}v} \leq 0$, and $N_{c}$ denotes the number of complex eigenvalues of $\mathcal{L}$ in the first quadrant.
    \par
    We will now show that $n(L_{-} + \delta L_{+}^{-1}) = n(L_{-})$. Due to the continuity of eigenvalues and the relative compactness of $L_{+}^{-1}$ with respect to $L_{-}$, we immediately have $n(L_{-}) \leq n(L_{-} + \delta L_{+}^{-1})$. On the other hand, it is possible that as $L_{-}$ is perturbed by $\delta L_{+}^{-1}$, a negative eigenvalue might bifurcate from the lower edge of the essential spectrum of $L_{-}$, which is at $0$. However, we will show that this does not occur. Consider the eigenvalue equation for $L_{-} + \delta L_{+}^{-1}$:
    \begin{equation}\label{eqn:eigval_eqn_Lm+dLp-1}
    (L_{-} + \delta L_{+}^{-1})w = \omega w.
    \end{equation}
    We can rewrite \eqref{eqn:eigval_eqn_Lm+dLp-1} as
    \begin{equation}
    (L + \delta B)w = \omega w,
    \end{equation}
    where
    \begin{align}
    L &= L_{-} + \delta L_{+,\infty}^{-1}, \quad B = L_{+}^{-1}\left(L_{+,\infty} - L_{+}\right)L_{+,\infty}^{-1} = L_{+}^{-1}\left(2|\bm{\varphi}| - 3|\bm{\varphi}|^{2}\right)L_{+,\infty}^{-1}.
    \end{align}
    Here, $L_{+,\infty} = -\dfrac{C}{2}\Delta - \mu$ represents the limiting form of the operator $L_{+}$ as $|x| \to \infty$. The operator $B$ is a relatively compact perturbation to $L$, and the essential spectrum of $L$ is bounded below by $\delta/|\mu| > 0$.
    \par
    By applying the theory of edge bifurcations \cite{kapitula_stanstede2004_eigenvalues_evans,kapitula_sandstede_2002_edge_bifurcation_evans}, any edge bifurcation that occurs when $\delta \neq 0$ and originates at the lower edge of the essential spectrum of $L$ would appear as
    \begin{equation}
        \omega(\delta) = \frac{\delta}{|\mu|} - a \delta^{2} + \mathcal{O}(\delta^{3}),
    \end{equation}
    where $a > 0$ is a constant. Thus, for sufficiently small positive $\delta$, $\omega(\delta) > 0$, which implies that the number of negative eigenvalues of $L_{-} + \delta L_{+}^{-1}$ remains unchanged. Consequently, $n(L_{-} + \delta L_{+}^{-1}) = n(L_{-})$ for small enough $\delta > 0$.
    \par
    Now, since $L_{-}\bm{\varphi} = 0$ and $\bm{\varphi}$ does not change signs on $\Z$, by discrete Sturm theory \cite{levy1992finite}, $L_{-}$ does not have a negative eigenvalue, therefore $n(L_{-}) = 0$. Equation \eqref{eqn:eigval_count_min} immediately gives that $N_{r}^{+} = N_{i}^{-} = N_{c} = 0$, thus $\mathcal{L}$ does not have any imaginary or complex eigenvalues. Finally, \eqref{eqn:eigval_count_plus} gives the relation
    \begin{equation} \label{eqn:eigval_count_plus_2}
        n(L_{+}) = N_{r}^{-},
    \end{equation}
    i.e., the number of real eigenvalues of $\mathcal{L}$ (with $\inp{v}{L_{+}^{-1}v} \leq 0$) is equal to the number of negative eigenvalues of $L_{+}$. $\blacksquare$
\end{proof}

Theorem \ref{thm:eigval_count} implies that any instability of the stationary solution $\bm{\varphi}$ can only occur due to pairs of real eigenvalues of $\mathcal{L}$. Moreover, the number of such pairs is given by the number of negative eigenvalues of $L_{+}$.

\par
It now remains to check whether the assumptions of Theorem \ref{thm:eigval_count} are satisfied by $L_{\pm}$ and $\bm{\varphi}$ when $0 < C \ll 1$. For small $C > 0$, the onsite and intersite Maxwell fronts remain positive at all lattice points, as can be seen from the perturbation analysis, and they satisfy the far-field conditions in the assumption of the theorem. At $C = 0$, $L_{+}$ reduces to a multiplicative operator where
\begin{equation}  \label{eqn:Lp_C=0}
    (L_{+}\psi)_{n} = 
    \begin{cases}
        -\mu_{M} \psi_{n}, &\text{if } \varphi_{n}^{(0)} =  \varphi^{\mathrm{low}} \text{ or } \varphi_{n}^{(0)} = \varphi^{\mathrm{up}}, \\
        \left(\frac{\mu_{M}}{2}\right) \psi_{n}, &\text{if } \varphi_{n}^{(0)} = \varphi^{\mathrm{mid}}.
    \end{cases}
\end{equation}
Thus, $L_{+}$ is invertible. Therefore, it remains invertible when the coupling is small but finite, and its kernel is empty. The kernel of $L_{-}$ is also empty, since $L_{-}\bm{\varphi} = 0$, $\bm{\varphi} \in \ell^{\infty}(\Z)$, and $\varphi \notin \ell^{2}(\Z)$. The onsite and intersite Maxwell fronts are also nonnegative on $\Z$, and thus they do not change signs on $\Z$. The assumptions of Theorem \ref{thm:eigval_count} are then satisfied, and we can apply the result of the theorem for the corresponding onsite and intersite Maxwell front solutions.
\par
We first discuss the stability of the intersite Maxwell front. At $C = 0$, the intersite Maxwell front contains no sites $n$ such that $\varphi_{n} = \varphi^{\mathrm{mid}}$. Thus, the operator $L_{+}$ reduces to $L_{+} = -\mu_{M} I$, where $I$ is the identity operator. Therefore, at $C = 0$, the spectrum of the operator $L_{+}$ corresponding to the intersite Maxwell front consists of only the point $-\mu_{M} > 0$. When the coupling $C$ is small but finite, the operator $L_{+}$ at $C = 0$ is perturbed by the discrete Laplacian operator $(C/2)\Delta$, and the spectrum splits from $-\mu_{M}$. By continuity of the eigenvalues in $C$, any eigenvalue which splits from $-\mu_{M}$ must remain $O(C)$-close to $-\mu_{M}$ for small enough $C$, and thus for sufficiently small $C$, $L_{+}$ has no negative eigenvalue, i.e., $n(L_{+}) = 0$. Therefore, we conclude that for $0 < C \ll 1$, the intersite Maxwell front is neutrally stable.
\par
Next, we turn to the case of the onsite Maxwell front. In the case of the onsite Maxwell front, at $C = 0$, the operator $L_{+}$ has one negative eigenvalue $\mu_{M}/2$, owing to the site $n = 0$, since $\varphi_{0} = \varphi^{\mathrm{mid}}$. As a result, when $C > 0$, this negative eigenvalue is perturbed and remains negative for small enough $C$, thus $n(L_{+}) = 1$. We can then conclude that the onsite Maxwell front is unstable for small $C > 0$ due to a pair of real eigenvalues of the stability operator $\mathcal{L}$. 
\par
Furthermore, we can approximate this pair of real eigenvalues of $\mathcal{L}$ by analyzing the eigenvalues of the operator $L_{+}L_{-}$ for small $C$. Observe that any real eigenvalues of $L_{+}L_{-}$ for small positive $C$ must arise from the zero eigenvalue at $C = 0$. Therefore, we will analyze the splitting of the zero eigenvalue of $L_{+}L_{-}$ for $0 < C \ll 1$. When $C = 0$, an eigenvector $v$ of $L_{+}L_{-}$ is supported on the sites $n \geq 0$, which motivates the following decaying ansatz for $v$ when $C > 0$:
\begin{equation} \label{eqn:ac_eigvec_v_ansatz}
    v_{n} = 
    \begin{cases}
        0, &n \leq  -1, \\
        1, &n = 0, \\
        \beta l^{n}, &n \geq 1,
    \end{cases}
\end{equation}
where $\beta$ is a to-be-determined constant and $l$ is the decay rate of $v$ such that $|l| < 1$, which will be determined. Now, expanding $L_{+}L_{-}$ and retaining terms up to order $O(C)$, we have
\begin{equation} \label{eqn:ac_LpLm_O(C)}
    \begin{split}
    (L_{+}L_{-}v)_{n} = &\left[\mu_{M}^{2} + 3\mu_{M}|\varphi_{n}| + (2-4\mu_{M})|\varphi_{n}|^{2} - 5 |\varphi_{n}|^{3} + 3|\varphi_{n}|^{4}\right]v_{n} \\
    &+ C\left[\mu_{M} \Delta v_{n} + \frac{1}{2} \Delta(|\varphi_{n}|v_{n}) + |\varphi_{n}|\Delta v_{n} - \frac{1}{2} \Delta\left(|\varphi_{n}|^{2}v_{n}\right) - \frac{3}{2}|\varphi_{n}|^{2}\Delta v_{n}\right].
    \end{split}
\end{equation}
Now, substituting the approximation for $\varphi$ up to $O(C)$ as in \eqref{eqn:ac_asymp_expansion_phi}, i.e., $\varphi = \varphi^{(0)} + C \varphi^{(1)}$, and the approximation of $L_{+}L_{-}$ from \eqref{eqn:ac_LpLm_O(C)}, while using the ansatz for $v$ as \eqref{eqn:ac_eigvec_v_ansatz} in the eigenvalue equation \eqref{eqn:eigval_eqn_LpLm}, we obtain three equations for $-\lambda^{2}$, each corresponding to the site $n = 0$, $n = 1$, and the far-field $n \gg 1$. Matching the equations from these sites would determine the constant $\beta$ and the decay rate $l$, from which we obtain a candidate for an eigenvector $v$. The calculations yield
\begin{equation}
    l = \dfrac{-\sqrt{17}+5}{2}, \quad \beta = \dfrac{2l}{2-l}.
\end{equation}
This gives the $O(C)$ approximation for $\lambda^{2}$ as
\begin{equation} \label{eqn:ac_lambda_approx}
    \lambda^{2} = \frac{{\left(3 \, \sqrt{17} + 13\right)}}{9{\left(\sqrt{17} + 5\right)}}C  + O(C^{2}).
\end{equation}

Another approach to arrive at the instability is to analyze the eigenvalue arising from the central site, $n = 0$, at which $\varphi_{n} = \varphi^{\mathrm{mid}}$. The reasoning for this is as follows. For small $C$, the sites $\varphi_{n}$ with $|n| \geq 2$ remain unchanged (up to first order in $C$), i.e., they take on the uniform background solutions $\varphi^{\mathrm{low}}$ or $\varphi^{\mathrm{up}}$, and thus they contribute to the continuous spectrum of $\mathcal{L}$. The sites $|n| = 1$ also remain $O(C)$-close to the uniform backgrounds. Therefore, if any instability occurs, it must be due to the central site $n = 0$. We shall refer to this as the one-site approximation.
\par
To derive the one-site approximation of the unstable eigenvalue, we write $L_{\pm}$ in a power series of $C$ as
\begin{align}
    L_{\pm} = L_{\pm}^{(0)} + C L_{\pm}^{(1)} + O(C^{2}),
\end{align}
where $L_{\pm}^{(0)}$ are the operators $L_{\pm}$ when $C = 0$, and the $O(C)$ terms $L_{\pm}^{(1)}$ are given by
\begin{align}
    L_{+}^{(1)} &= -\frac{1}{2}\Delta - 2 \varphi_{n}^{(1)} + 6\varphi_{n}^{(0)}\varphi_{n}^{(1)}, \\
    L_{-}^{(1)} &= -\frac{1}{2}\Delta - \varphi_{n}^{(1)} + 2\varphi_{n}^{(0)}\varphi_{n}^{(1)}.
\end{align}
When $C = 0$, the central site corresponds to the zero eigenvalue of $L_{+}L_{-}$, so we perturb off of this zero eigenvalue. We write the perturbation expansion of the eigenvalues and eigenvectors as
\begin{align}
    -\lambda^{2} = \theta &= \theta^{(0)} + C\theta^{(1)} + O(C^{2}), \\
    v &= v^{(0)} + C v^{(1)} + O(C^{2}),
\end{align}
where $\theta^{(0)} = 0$ and the leading-order eigenvector $v^{(0)}$ is taken to be active only at the central site, i.e., $v^{(0)}_{0} = 1$ and $v^{(0)}_{n} = 0$ for $n \geq 0$, following the preceding discussion. The equation at $O(1)$ is then automatically satisfied by the choice of $v^{(0)}$. At $O(C)$, we obtain the equation
\begin{equation}
    \left[L_{+}^{(0)}L_{-}^{(0)} - \theta^{(0)}\right]v_{1} = \left[\theta^{(1)} - L_{+}^{(1)}L_{-}^{(0)} - L_{+}^{(0)}L_{-}^{(1)}\right]v_{0} = f.
\end{equation}
By the Fredholm alternative, the solvability of this equation is given by $\inp{f}{v_{0}^{(0)}} = 0$, from which we obtain the $O(C)$ correction term of the eigenvalue as $\theta^{(1)} = -1/9$. Thus, the one-site approximation of the unstable eigenvalue is
\begin{align}\label{eqn:ac lambda approx one site approx}
    \lambda^{2} = \frac{1}{9}C + O(C^{2}).
\end{align}

The qualitative stability result for small $C$ extends similarly in the case of general front configurations with length $N \geq 2$. Consider a general stationary front $\bm{\varphi}$ with code $\mathcal{C}$. As the nontrivial uniform steady states $\varphi^{\mathrm{mid}}$ and $\varphi^{\mathrm{up}}$ \eqref{eqn:2-3_uniform_solutions} are positive for $\mu \in (-1/4,0)$, the perturbation analysis shows that for small enough $0 < C \ll 1$, $\bm{\varphi}$ remains nonnegative and they satisfy the far-field conditions of Theorem \ref{thm:eigval_count}. It remains to check the assumption on the linearized operators $L_{\pm}$. $L_{-}$ has empty kernel since $\bm{\varphi} \in \ell^{\infty}(\Z)$ and it is not in $\ell^{2}(\Z)$. The operator $L_{+}$ acts as a multiplicative operator
\begin{equation}  \label{eqn:Lp_C=0_general_mu}
    (L_{+}\psi)_{n} = 
    \begin{cases}
        -\mu \psi_{n}, &\text{if } \varphi_{n}^{(0)} =  \varphi^{\mathrm{low}}, \\
        \left[\varphi_{n}^{(0)}F'\left(|\varphi_{n}^{(0)}|\right)\right] \psi_{n}, &\text{if } \varphi_{n}^{(0)} = \varphi^{\mathrm{mid}} \text{ or } \varphi_{n}^{(0)} = \varphi^{\mathrm{up}},
    \end{cases}
\end{equation}
at $C = 0$. Since $\varphi F'(|\varphi|) \neq 0$ for the nontrivial uniform states, $L_{+}$ is an invertible operator at $C = 0$ and it remains invertible for small $C$. Moreover, since $\varphi^{\mathrm{mid}}F'(|\varphi^{\mathrm{mid}}|) < 0$ and $\varphi^{\mathrm{up}}F'(|\varphi^{\mathrm{up}}|) > 0$ for $\mu \in (-1/4,0)$, we infer that the number of negative eigenvalues of $L_{+}$ at $C=0$ is given by the number of sites taking on the value $\varphi^{\mathrm{mid}}$. By persistence of eigenvalues, this number remains the same for small enough $C$. Therefore, the number of pairs of real unstable eigenvalues of $\mathcal{L}$ is given by
\begin{equation} \label{eqn:number of unstable eigenvalues}
    n(L_{+}) = \#\left\{n \in \Z: \varphi_{n}^{(0)} = \varphi^{\mathrm{mid}}\right\},
\end{equation}
in other words, it is the number of times the character '$m$' appears in the code $\mathcal{C}$. 
Note that $n(L_{+})$ is always finite for any front configuration since the front can only take on the value $\varphi^{\mathrm{mid}}$ in its center $\bm{\tilde{\varphi}}$ which is of finite length $N$. In light of this analysis, we see that the stability of a front (for small $C$) inherits that of the unstable uniform state $\varphi^{\mathrm{mid}}$ if the anticontinuum front (its $C=0$ configuration) takes this value on some the lattice sites. The approximation for the eigenvector and eigenvalues can be done similarly as in the onsite and intersite case.


The asymptotic results \eqref{eqn:ac_lambda_approx} and \eqref{eqn:ac lambda approx one site approx} are compared with the numerics in Figure \ref{fig:ac_eigval_compare}, yielding a fairly good agreement for small values of $C$. 
While both approximations \eqref{eqn:ac_lambda_approx} and \eqref{eqn:ac lambda approx one site approx} agree with the numerical calculations asymptotically as $C \to 0$, the asymptotic approximation \eqref{eqn:ac_lambda_approx} provides a more accurate fit with the numerical results for slightly larger values of $C$.
\begin{figure}[htbp]
    \centering
    \includegraphics[width=0.6\linewidth]{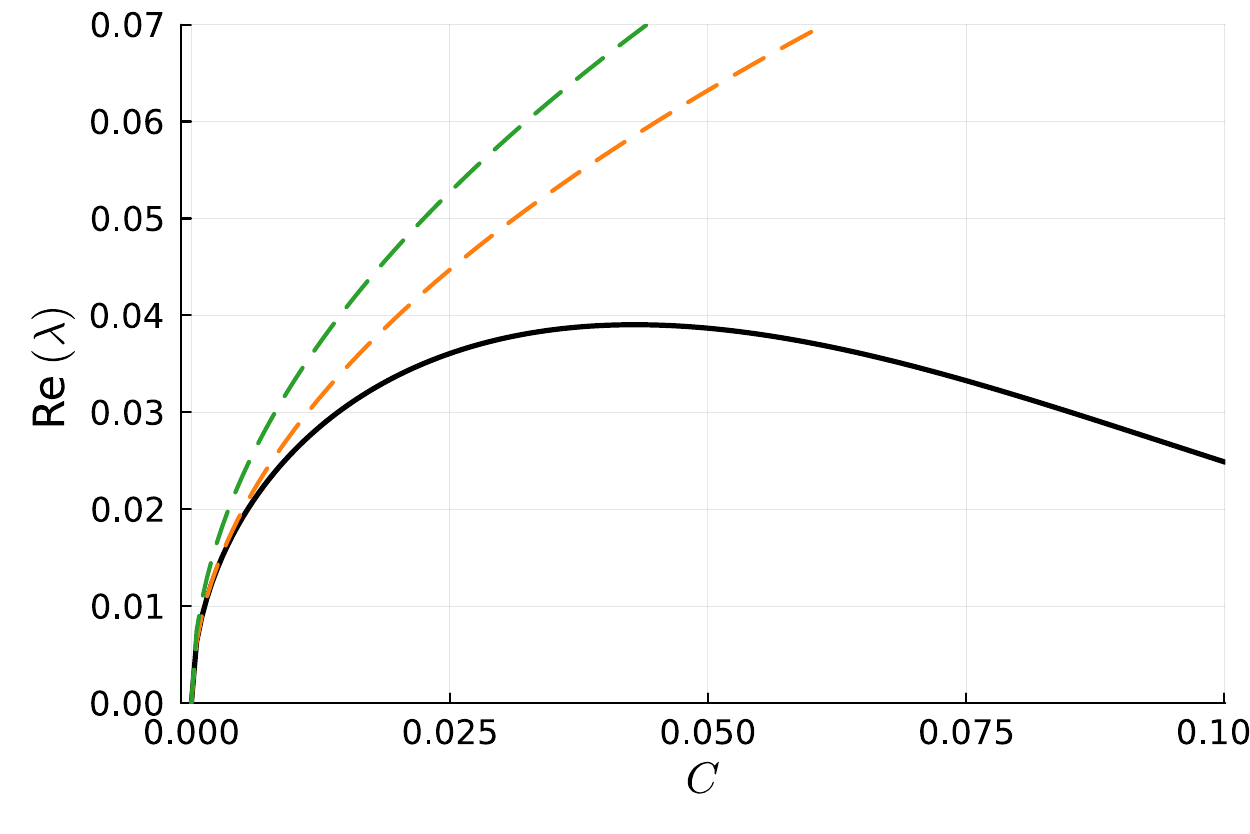}
    \caption{Real part of the eigenvalue of $\mathcal{L}$ for the onsite quadratic-cubic Maxwell front for $C \in [0,0.1]$. The black solid line shows the numerically computed eigenvalue, the orange dashed line shows the asymptotic approximation \eqref{eqn:ac_lambda_approx}, and the green dashed line shows the one-site approximation \eqref{eqn:ac lambda approx one site approx}.}
    \label{fig:ac_eigval_compare}
\end{figure}
As $C$ gets larger, the real part of the unstable eigenvalue reaches a maximum at a certain value of $C$, after which it decreases as $C$ increases (see section \ref{sec3.3} for further elaboration).


\subsection{Continuum Limit} \label{sec3.2}
As noted in the beginning of this section, in the continuum limit, Maxwell fronts exist only at the Maxwell point $\mu = \mu_{M}$, thus we will focus on the case $\mu = \mu_{M}$ in this and the subsequent sections. We will now analyze the stability of Maxwell fronts in the continuum limit, where the coupling between sites is large, i.e., $C \gg 1$. In this limit, it is instructive to introduce a small parameter $h$ and rewrite the coupling parameter $C$ as $C = 2/h^{2}$ and consider the equivalent case of $h \to 0$. With this, equation \eqref{eqn:quad_cubic_dnls} becomes
\begin{equation}
    i \dot{\phi}_{n} = -\dfrac{1}{h^{2}}\Delta \phi_{n} - \mu \phi_{n} - |\phi_{n}|\phi_{n} + |\phi_{n}|^{2}\phi_{n}.
\end{equation}
As the coupling between lattice sites is large, the variation of $\phi$ between neighboring sites is expected to be small, implying that $\phi$ varies slowly with respect to $n$. To this end, we introduce the slow scale $x = hn$ and assume that $\phi$ depends on this slow scale, i.e. $\phi_{n} \equiv \phi(x,t)$. Under the slow scale, the Laplacian term can be approximated as
\begin{equation}
    \begin{split}
    \Delta \phi_{n} &= \phi(x+h,t) - 2\phi(x,t) + \phi(x-h,t) \\
    &= 2\sum_{k \geq 1} \dfrac{h^{2k}}{(2k)!} \partial_{x}^{2k}\phi.
    \end{split}
\end{equation}
\subsubsection{\texorpdfstring{$C \to \infty$ $(h = 0)$}{h = 0}} \label{sec3.2.1}
By writing the Laplacian in a power series of $h$, we can express equation \eqref{eqn:quad_cubic_dnls} as
\begin{equation} \label{eqn:c_quad_cubic_dnls_C_large}
    i \partial_{t}\phi = -2\sum_{k\geq 1} \dfrac{h^{2k-2}}{(2k)!}\partial_{x}^{2k}\phi - \mu_{M} \phi - |\phi|\phi + |\phi|^{2}\phi.
\end{equation}
In the continuum limit ($C \to \infty$ or equivalently $h = 0$), equation \eqref{eqn:c_quad_cubic_dnls_C_large} becomes the continuous equation
\begin{equation} \label{eqn:cont_quad_cubic_nls}
    i \phi_{t} = -\partial_{x}^{2}\phi - \mu_{M} \phi - |\phi|\phi + |\phi|^{2}\phi.
\end{equation}
Correspondingly, the stationary equation \eqref{eqn:quad_cubic_dnls_stationary} becomes \eqref{eqn:2-3 continuum eqn} and we have the continuum analogue to the Maxwell front given by \eqref{eqn:cont_front_solution} where $x_{0}$ is an arbitrary shifting constant.
The fact that the shifting constant $x_{0}$ is freely chosen reflects the translational invariance of the continuum equation \eqref{eqn:2-3 continuum eqn}.
Thus, in the continuum limit, the Maxwell front is not limited to the onsite and intersite configurations (corresponding to $x_{0} = 0$ and $x_{0} = 1/2$, respectively), but exists for any $x_{0} \in \mathbb{R}$, leading to an infinite number of Maxwell fronts.
\par
In this limit, the linearized operators, $L_{+}$ and $L_{-}$, simplify to the following differential operators:
\begin{align}
    L_{-} &= -\partial_{x}^{2} - \mu_{M} - |\varphi| + |\varphi|^{2},  \label{eqn:cont_Lp} \\
    L_{+} &= -\partial_{x}^{2} - \mu_{M} - 2|\varphi| + 3|\varphi|^{2}.
\end{align}
A far-field analysis of each differential operator, along with an application of Weyl's Essential Spectrum Lemma, shows that the essential spectrum of $L_{+}$ and $L_{-}$ is bounded below by $|\mu_{M}|$ and $0$, respectively. Specifically, their essential spectra are given by $\sigma_{e}(L_{+}) = \left[|\mu_{M}|, \infty\right)$ and $\sigma_{e}(L_{-}) = [0, \infty)$. Furthermore, since $L_{-}\varphi = 0$, differentiating with respect to $x$ gives $L_{+}\varphi' = 0$. Thus, $L_{+}$ has a zero eigenvalue corresponding to the strictly positive eigenfunction $\varphi'(x)$. Consequently, by Sturm-Liouville theory, $L_{+}$ has no negative eigenvalues. Theorem \ref{thm:eigval_count} then implies that the stability operator $\mathcal{L}$ has no real eigenvalues. Therefore, the Maxwell front is neutrally stable in the continuum limit.

\subsubsection{\texorpdfstring{$C \gg 1$ $(0 < h \ll 1)$}{h small}} \label{sec3.2.2}
When the coupling is large but finite ($C \gg 1$ or equivalently $0 < h \ll 1$), using the method of exponential asymptotics \cite{KING_CHAPMAN_2001}, it can be shown that the translational invariance of the continuum limit is broken, thus permitting only the onsite and intersite configurations of the Maxwell fronts.
\par
The analysis of the stationary solution in this case proceeds as follows. We express the discrete Laplacian term in powers of $h^{2j}$, as before, and consider the discrete stationary equation:
\begin{equation} \label{eqn:cont_stationary_equation_discrete}
    0 = -2\sum_{k \geq 1} \dfrac{h^{2k-2}}{(2k)!} \partial_{x}^{2k}\varphi - \mu_{M}\varphi - |\varphi|\varphi + |\varphi|^{2}\varphi.
\end{equation}
Expanding $\varphi(x)$ in an asymptotic power series in $h$ as
\begin{equation} \label{eqn:cont_asymp_expansion_varphi}
    \varphi(x) \sim \sum_{j \geq 0} h^{2j} \varphi_{j}(x) \quad (h \ll 1),
\end{equation}
gives the continuum one-front solution \eqref{eqn:cont_front_solution} at leading order:
\begin{equation} \label{eqn:cont_varphi0_leading_order}
    \varphi_{0}(x) = \frac{2}{3}\left(1 + e^{-\sqrt{2}(x-x_{0})/3}\right)^{-1}.
\end{equation}
Here, $x_{0} = h n_{0}$ is a constant that will be determined through the exponential asymptotic analysis. We will see that only two values of $n_{0}$ admit one-front solutions in this case. For brevity, we will write the equations in the subsequent calculations in the shifted variable $\tilde{x} = x - x_{0}$. Under this shifted variable, we have $\partial_{\tilde{x}} = \partial_{x}$.
\par
An important observation is that $\varphi_{0}(\tilde{x})$ possesses poles at $\zeta_{k} = i3\pi(1 + 2k)/\sqrt{2}$ $(k \in \mathbb{Z})$, where $\varphi_{0}(\tilde{x}) \sim \sqrt{2}(\tilde{x}-\zeta_{k})^{-1}$ as $\tilde{x} \to \zeta_{k}$. These singularities become more pronounced in the late-order terms, $\varphi_{j}$, of the expansion \eqref{eqn:cont_asymp_expansion_varphi}, in a factorial over power form, owing to the dominant contribution of the poles nearest to the real axis, $\zeta_{0}$ and $\zeta_{-1} = \overline{\zeta}_{0}$. Specifically, the expansion takes the form:
\begin{equation} \label{eqn:cont_late_order_varphij}
    \varphi_{j}(\tilde{x}) \sim \dfrac{(-1)^{j} \Gamma(2j+4) \Lambda_{1} G(\tilde{x})}{\chi^{2j+\beta}} \quad (j \gg 1),
\end{equation}
where $\chi = 2\pi(\tilde{x}-\zeta)$, $G(\tilde{x}) = \varphi_{0}'(\tilde{x}) \int_{\zeta}^{\tilde{x}} [\varphi_{0}'(s)]^{-2} \, ds$, and $\zeta$ is taken to be either of the dominant poles, $\zeta_{0}$ or $\zeta_{-1}$. The constant $\Lambda_{1}$ is determined by matching an inner solution in a region of order $O(h)$ near $\zeta$ with the outer solution of \eqref{eqn:cont_late_order_varphij} and is found to be approximately $\Lambda_{1} \approx -2533$.
    \par
The factorial-over-power divergence of $\varphi_{j}$ triggers the activation of an exponentially small remainder term due to the Stokes phenomenon \cite{dingle1973asymptotic}, which occurs across the Stokes line, where
\begin{equation}
    \text{Im}\ \chi = 0, \quad \text{Re}\ \chi \geq 0.
\end{equation}
In this case, the Stokes lines are given by $\text{arg}(\tilde{x}-\zeta_{0}) = -\pi/2$ and $\text{arg}(\tilde{x}-\zeta_{-1}) = \pi/2$. The analysis of the Stokes phenomenon can be carried out by considering the optimally truncated asymptotic expansion of $\varphi$:
\begin{equation} \label{eqn:cont_trunc_expansion}
    \varphi(\tilde{x}) = \sum_{j=0}^{N-1} h^{2j} \varphi_{j}(\tilde{x}) + R_{N}(\tilde{x}),
\end{equation}
where $N$ is chosen such that the approximation error of the series is minimized, and the remainder $R_{N}$ is exponentially small in $h$. Upon substituting \eqref{eqn:cont_trunc_expansion} into \eqref{eqn:cont_stationary_equation_discrete} and using the late-order divergence of $\varphi_{j}$ from \eqref{eqn:cont_late_order_varphij}, we obtain an equation for the remainder $R_{N}$. Solving the equation for $R_{N}$ while considering the behavior of $R_{N}$ in a neighborhood of the pole $\zeta$ such that $\text{arg}(\tilde{x}-\zeta) = O(\sqrt{h})$, it is found that the exponentially small remainder term $R_{N}$ is "switched on" as the Stokes line is crossed from left to right ($\text{Re}(\tilde{x}) < 0$ to $\text{Re}(\tilde{x}) > 0$). 

Taking into account the contributions from $\zeta_{0}$ and $\zeta_{-1}$, it is found that for real $\tilde{x} > 0$, the remainder behaves as
\begin{equation}
    R_{N}(\tilde{x}) \sim -2\pi \Lambda_{1} h^{-4} e^{-3\sqrt{2} \pi^2 / h} \left[ \text{Re}(G(\tilde{x})) \sin\left(\frac{2\pi \tilde{x}}{h}\right) - \text{Im}(G(\tilde{x})) \cos\left(\frac{2\pi \tilde{x}}{h}\right) \right], \quad (h \ll 1),
\end{equation}
which is "switched on" across the Stokes line.
\par
Taking this remainder term into consideration for the solution $\varphi$, we have
\begin{equation} \label{eqn:varphi_expansion_truncated_with_remainder}
    \varphi(\tilde{x}) \sim \varphi_{0}(\tilde{x}) + \sum_{j=1}^{N-1} h^{2j} \varphi_{j}(\tilde{x}) + R_{N}(\tilde{x}), \quad (h \ll 1, \tilde{x} > 0).
\end{equation}
We are now in a position to determine the shifting constant $x_{0} = h n_{0}$ that admits one-front solutions. For $\varphi(\tilde{x})$ to be a solution of the discrete equation \eqref{eqn:cont_stationary_equation_discrete}, the remainder $R_{N}(\tilde{x})$ must remain bounded. However, $R_{N}(\tilde{x})$ is not necessarily bounded for any $x_{0}$, since $G(\tilde{x})$ is real and grows unbounded as $\tilde{x} \to \infty$:
\begin{equation} \label{eqn:cont_G(x)_large_x}
    G(\tilde{x}) \sim \dfrac{27}{8} e^{\sqrt{2} \tilde{x}/3} \quad (\tilde{x} \to \infty).
\end{equation}
Thus, as $\tilde{x} \to \infty$, $R_{N}(\tilde{x})$ behaves asymptotically as
\begin{equation}
    R_{N}(\tilde{x}) \sim \dfrac{27\pi}{4} \Lambda_{1} h^{-4} e^{-3\sqrt{2} \pi^2 / h} \sin(2\pi n_{0}) e^{\sqrt{2} \tilde{x}/3}.
\end{equation}
It is now evident that for $R_{N}(\tilde{x})$ to remain bounded as $\tilde{x} \to \infty$, $n_{0}$ can only take two values (modulo integer shifts), namely $n_{0} = 0$ and $n_{0} = 1/2$, corresponding to the onsite and intersite stationary Maxwell fronts, respectively.
\par
Next, we turn to the spectrum of $L_{+}$ and $L_{-}$ in the $C \gg 1$ $(h \ll 1)$ case. Weyl's Essential Spectrum Lemma gives that the essential spectrum of $L_{+}$ and $L_{-}$ are also bounded below by $|\mu_{M}|$ and $0$. However, we shall see that in this case, $L_{+}$ has no zero eigenvalue. In fact, the zero eigenvalue of $L_{+}$ in the continuum limit $(h = 0)$ bifurcates as $0 < h \ll 1$, and the bifurcation depends on the configuration of Maxwell fronts. This bifurcation occurs because the translational invariance of the solution $\varphi(x)$ is broken. In the continuum limit $(h = 0)$, the translational invariance of the Maxwell front allows us to obtain $L_{+} \varphi'(x) = 0$ by differentiating $L_{-} \varphi(x) = 0$ with respect to $x$. In the $0 < h \ll 1$ case, we consider the stationary solution $\varphi(\tilde{x})$ and the fact that $L_{-} \varphi(\tilde{x}) = 0$ for $n_{0}$ taking either of the two values $0$ or $1/2$. Even though we obtain that $L_{+} \varphi' = 0$, notice that $\varphi' \notin \ell^{2}(\Z)$, since, due to the remainder term, as $\tilde{x} \to \infty$, $\varphi'(\tilde{x})$ has an exponentially growing tail as
\begin{equation}
    \begin{split}
        \varphi'(\tilde{x}) &= \varphi_{0}'(\tilde{x}) + R_{N}'(\tilde{x}) \\
        &\sim -\dfrac{27\pi^{2}}{2} \Lambda_{1} h^{-5} e^{-3\sqrt{2} \pi^2 / h} e^{\sqrt{2} \tilde{x}/3} \cos(2\pi n_{0}) \quad (\tilde{x} \to \infty),
    \end{split}
\end{equation}
for both the onsite and intersite Maxwell fronts. The fact that $\varphi'(\tilde{x})$ has an exponentially growing tail plays a crucial role in analyzing the perturbation of the zero eigenvalue of $L_{+}$ from the continuum limit. To distinguish between the operator $L_{+}$ in the continuum case and the large coupling case, we parameterize $L_{+}$ with $h$ as $L_{+} = L_{+}(h)$ and write $L_{+}(0)$ to denote the continuum operator, as in \eqref{eqn:cont_Lp}, and $L_{+}(h)$ with $0 < h \ll 1$ to denote the operator in the large coupling case.

    \par
To analyze the bifurcation of the zero eigenvalue of $L_{+}(h)$ from the continuum case $(h=0)$ to the large coupling case $(0 < h \ll 1)$, we consider the eigenvalue equation
\begin{equation} \label{eqn:cont_eigval_eqn_Lp}
    L_{+}(h) w = \alpha w.
\end{equation}
Since we are considering the bifurcation of the zero eigenvalue, which is near the essential spectrum of $L_{+}$ in the continuum limit, we expect that the bifurcated eigenvalue is small in magnitude, i.e., $|\alpha| \ll 1$. Thus, we write $w$ in an asymptotic expansion in $\alpha$ as
\begin{equation}
    w = w_{0} + \alpha w_{1} + \dots \quad (\alpha \ll 1).
\end{equation}
At leading order, this gives the equation
\begin{equation}
    L_{+}(h) w_{0} = 0,
\end{equation}
which is satisfied by $w_{0} = \varphi'(\tilde{x})$. As discussed above, $w_{0}$ has a growing tail as $\tilde{x} \to \infty$, given by
\begin{equation}
    w_{0} \sim -\dfrac{27\pi^{2}}{2} \Lambda_{1} h^{-5} e^{-3\sqrt{2} \pi^{2} / h} e^{\sqrt{2} \tilde{x} / 3} \cos(2\pi n_{0}).
\end{equation}
This growing tail will be counterbalanced by the solution at order $O(\alpha)$. At $O(\alpha)$, we have the equation
\begin{equation}
    L_{+}(h) w_{1} = w_{0}.
\end{equation}
Since as $h \to 0$, $\varphi(\tilde{x}) \to \varphi_{0}(\tilde{x})$, we have asymptotically
\begin{equation}
    L_{+}(0) w_{1} \sim \varphi_{0}'(\tilde{x}).
\end{equation}
Solving this equation yields $w_{1} = \varphi_{0}'(\tilde{x}) \nu(\tilde{x})$, where $\nu(\tilde{x})$ satisfies
\begin{equation}
    \nu'(\tilde{x}) = -\frac{1}{[\varphi_{0}'(\tilde{x})]^{2}} \int_{-\infty}^{\tilde{x}} [\varphi_{0}'(s)]^{2} \, ds.
\end{equation}
As $\tilde{x} \to \infty$, $\nu(\tilde{x})$ behaves asymptotically as
\begin{equation}
    \nu'(\tilde{x}) \sim -\frac{1}{2\sqrt{2}} e^{2\sqrt{2} \tilde{x} / 3},
\end{equation}
thus, $w_{1}$ exhibits a growing tail as
\begin{equation}
    w_{1} \sim -\frac{1}{6\sqrt{2}} e^{\sqrt{2} \tilde{x} / 3} \quad (\tilde{x} \to \infty).
\end{equation}
Taking into account contributions up to $O(\alpha)$, we have as $\tilde{x} \to \infty$,
\begin{equation}
    w \sim w_{0} + \alpha w_{1} \sim \left[-\dfrac{27\pi^{2}}{2} \Lambda_{1} h^{-5} e^{-3\sqrt{2} \pi^{2} / h} \cos(2\pi n_{0}) - \dfrac{\alpha}{6\sqrt{2}} \right] e^{\sqrt{2} \tilde{x} / 3}.
\end{equation}
Eliminating the growing tail yields the asymptotic formula for $\alpha$ as
\begin{equation}
    \alpha \sim -81\sqrt{2}\pi^{2} h^{-5} e^{-3\sqrt{2} \pi^{2} / h} |\Lambda_{1}| \cos(2\pi n_{0}) \quad (0 < h \ll 1).
\end{equation}
We see that 
the zero eigenvalue of $L_{+}(0)$ bifurcates to an exponentially small (in $h$) eigenvalue $\alpha$, and this bifurcation depends on the configuration of the Maxwell front. For the onsite Maxwell front $(n_{0} = 0)$, the zero eigenvalue bifurcates to a negative eigenvalue $(\alpha < 0)$, whereas for the intersite Maxwell front $(n_{0} = 1/2)$, it bifurcates to a positive eigenvalue $(\alpha > 0)$. Moreover, in the case of the intersite Maxwell front, due to the positivity of the eigenfunction $w$, it follows from discrete Sturm theory that $L_{+}$ has no eigenvalue smaller than $\alpha$. Therefore, we conclude that $L_{+}(h)$ has no negative eigenvalue in the case of the intersite Maxwell front.
\par
Now, the assumptions of Theorem \ref{thm:eigval_count} are satisfied, since $L_{\pm}$ again have empty kernels in $\ell^{2}(\Z)$ (due to the growing tail of $\varphi'(\tilde{x})$ and the fact that $\varphi(\tilde{x}) \notin \ell^{2}(\Z)$). Also, by the positivity of $\varphi$, it does not change signs on $\Z$, so we may apply Theorem \ref{thm:eigval_count}. We deduce that in the large coupling limit, the stability of the Maxwell fronts is the same as in the small coupling limit, i.e., the onsite Maxwell front is unstable due to the presence of the negative eigenvalue of $L_{+}$, and the intersite Maxwell front is neutrally stable.
\par
It is interesting to note the contrast between the stabilities of the dark solitons of the defocusing cubic DNLS (with $F(|\phi_{n}|) = |\phi_{n}|^{2}$ in \eqref{eqn:general_dnls}) and those of the quadratic-cubic Maxwell fronts. In the former case, the intersite configuration is always unstable for $C > 0$ due to a pair of real eigenvalues \cite{adriano2025exponential}, whereas the analysis above shows that in the latter, the intersite quadratic-cubic Maxwell front is always neutrally stable for $C > 0$. The onsite configurations of the front solutions also exhibit different stability characteristics. The onsite DNLS dark soliton is initially stable near the anticontinuum limit, up until a certain critical value of $C = C_{\text{cr}} \approx 0.01547$, where a Hamiltonian-Hopf bifurcation occurs due to a collision of an imaginary eigenvalue and an eigenvalue bifurcating from the band edge of the continuous spectrum, leading to oscillatory instability (instability due to a complex eigenvalue $\lambda$ in the first quadrant) \cite{PhysRevE.75.066608,PhysRevLett.82.85}. This instability persists as $C$ increases up to the continuum limit. The onsite quadratic-cubic Maxwell front, however, is always unstable for $C > 0$, due to a pair of real eigenvalues.

\subsection{Numerical Results}\label{sec3.3}
In this section, we conduct a numerical study on the stationary solutions of \eqref{eqn:quad_cubic_dnls} and their spectra. This is done by solving the stationary equation \eqref{eqn:quad_cubic_dnls_stationary} and the associated eigenvalue problem \eqref{eqn:eigval_eqn} on a truncated domain of size $N_{x}$. At the boundaries, we impose the homogeneous Neumann boundary condition on $\varphi_{n}$. For most of the computations, the size of the domain is chosen as $N_{x} = 1001$. The only exception is for the computation of the onsite Maxwell fronts for $C \geq 0.4$, where a domain of size $N_{x} = 30001$ is used. The computations are done in double precision. 

The fronts are computed initially at the anticontinuum limit $(C = 0)$ and as $C$ varies, the solutions are computed using numerical continuation. The results are illustrated in a bifurcation diagram where the coupling parameter $C$ is plotted against the squared $\ell^{2}(\Z)$ norm of the first difference of the front $\bm{\varphi}$, that is
\begin{equation}\label{eqn:N}
    \mathcal{N} = \sum_{n=-\infty}^{\infty}|\varphi_{n}-\varphi_{n-1}|^{2}.
\end{equation}

Figure \ref{fig:2-3_bifurcation_diagram_mu=maxwell} shows the bifurcation diagrams of Maxwell fronts with code lengths $N \leq 3$. It can be observed that as $C$ grows, the onsite $(0)$ and intersite $(m)$ branches becomes closer as their profiles converge to the continuum Maxwell front \eqref{eqn:cont_varphi0_leading_order}. However, they are still separated from each other. On the other hand, branches from fronts with codes of length $2$ and $3$ undergo fold bifurcations and connect with other branches of fronts (also of lengths $\geq 1$). In the case of codes of length 2 (the left panel of Fig. \ref{fig:2-3_bifurcation_diagram_mu=maxwell}), the branch $(mm)$ connects with the branch $(ul)$, along which the front's stability changes from unstable to stable. The branches $(ml)$ and $(um)$ are also connected, and they seem to appear as a single branch in the bifurcation diagram \ref{fig:2-3_bifurcation_diagram_mu=maxwell}. Along this branch, a change of stability also occurs, albeit within at a very small range of values of $\mathcal{N}$, as shown in the inset of Fig. \ref{subfig:2-3_bifurcation_length2}. For codes of length $3$ (the right panel of Fig. \ref{fig:2-3_bifurcation_diagram_mu=maxwell}), the picture becomes more complicated, as the bifurcation diagram shows intersections and overlaps between various branches. We note that these intersections and overlaps do not necessarily imply any bifurcation between these branches. This occurs mainly due to the fact that the bifurcation diagram only partially captures the full picture of the bifurcation of the fronts by projecting it onto the $\mathcal{N}-C$ plane. Regardless, we see that the codes of length $3$ undergo similar fold bifurcations as the length $2$ codes and they do not persist for values of $C$ greater than $0.05$. We note that as predicted through the analytical results, front configurations containing the character '$m$' in their codes are unstable as they bifurcate from $C = 0$, whereas those consisting solely of '$u$' and '$l$' characters are stable. Examples of the profile and spectrum of an unstable and stable fronts of length $\geq 2$ are illustrated in Fig. \ref{fig:2-3_fronts_length_23}. The unstable front branching from the code $(mmm)$ has 3 pairs of unstable eigenvalues, in agreement with the analytical prediction \eqref{eqn:number of unstable eigenvalues}. 
\begin{figure}[tbhp]
    \centering
    \subfloat[Codes of length $0,1$, and $2$]{\includegraphics[scale=0.335]{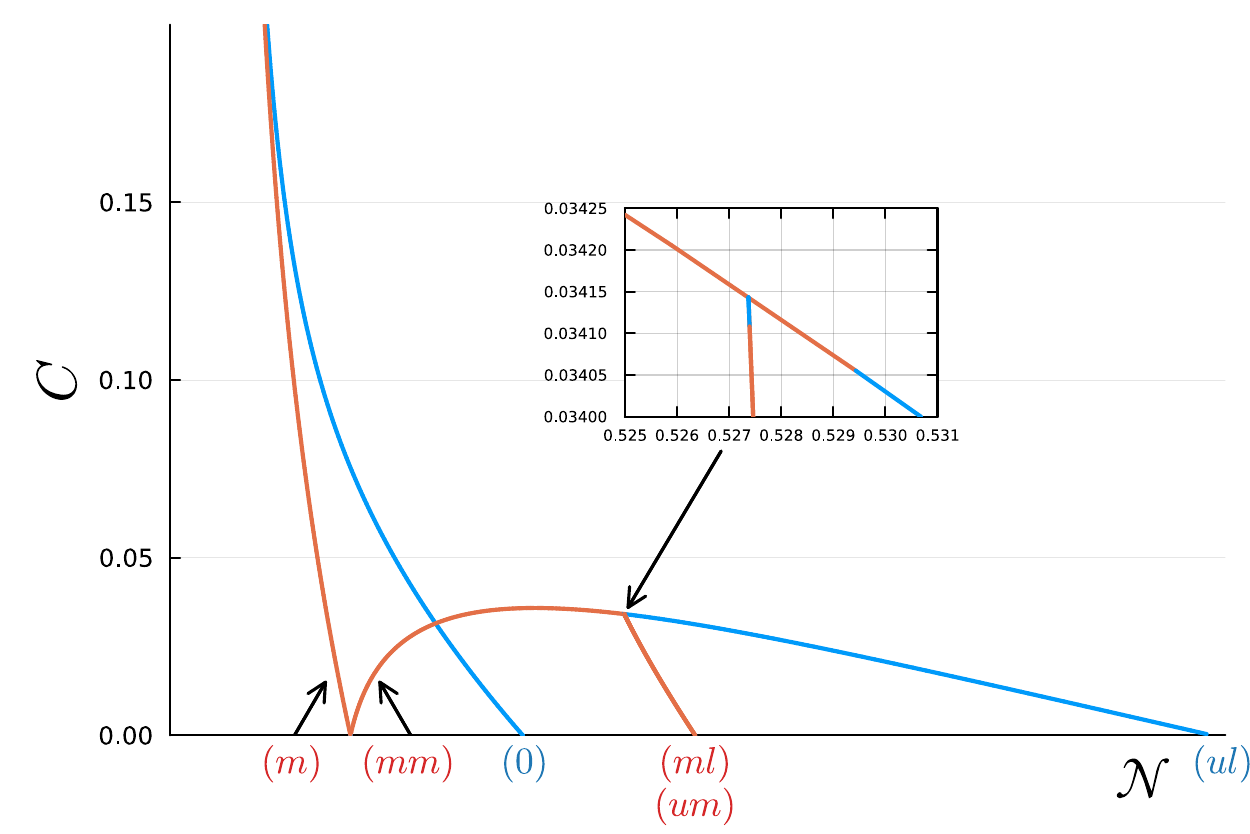}\label{subfig:2-3_bifurcation_length2}}
    \subfloat[Codes of length $0,1$, and $3$]{\includegraphics[scale=0.335]{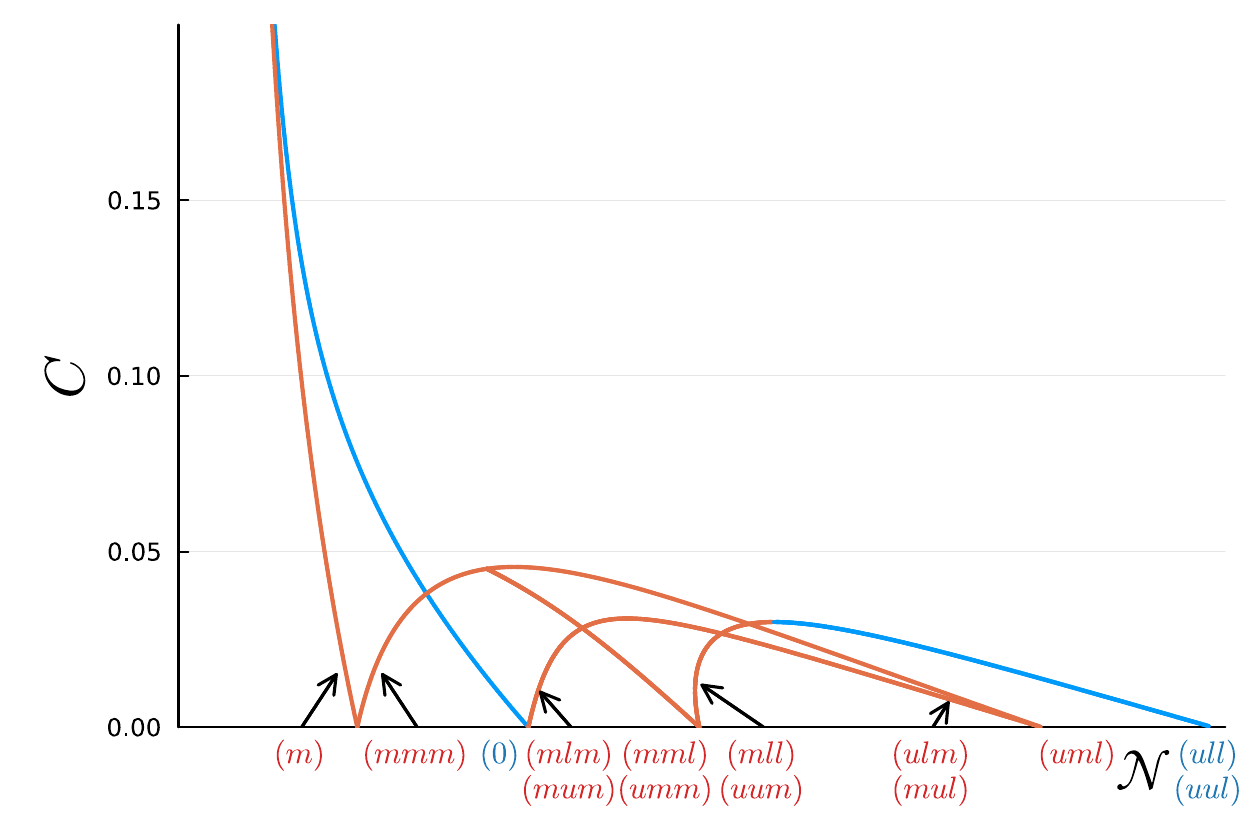}\label{subfig:2-3_bifurcation_length3}}
    \caption{Bifurcation diagrams of quadratic-cubic Maxwell fronts with codes of length $0,1,2$, and $3$. Blue lines represent stable fronts and red lines represent unstable fronts.}
    \label{fig:2-3_bifurcation_diagram_mu=maxwell}
\end{figure}

\begin{figure}[tbhp]
     \centering
     \begin{subfigure}{1\textwidth}
         \includegraphics[width = \textwidth]{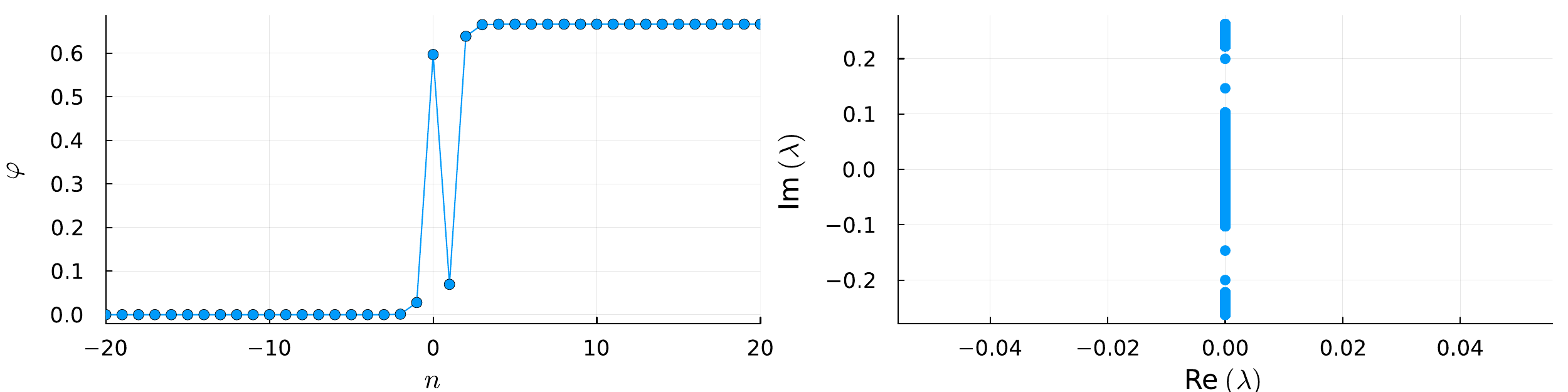}
         \caption{$(ul)$ Maxwell front}
     \end{subfigure}
     \begin{subfigure}{1\textwidth}
         \includegraphics[width = \textwidth]{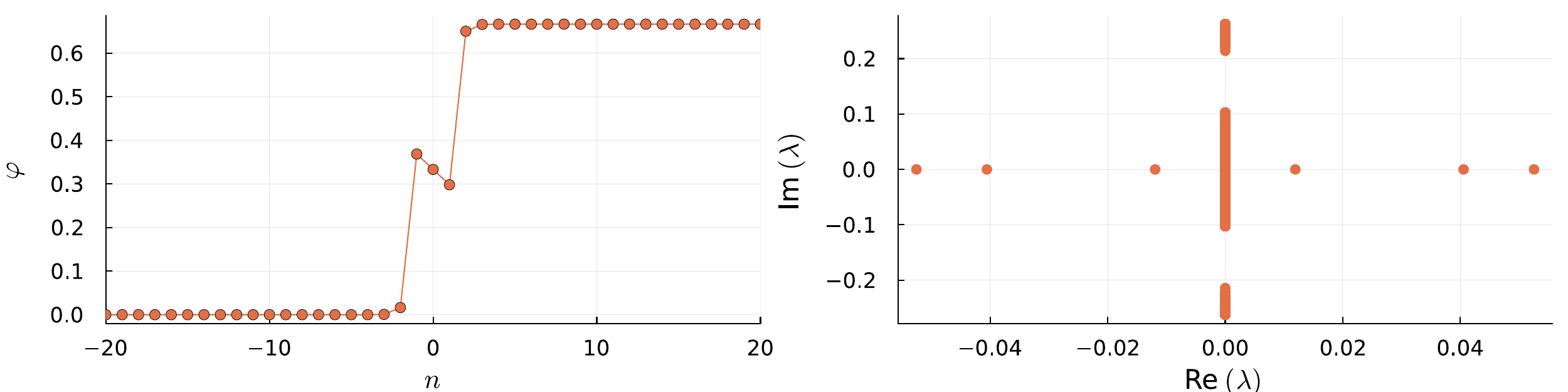}
         \caption{$(mmm)$ Maxwell front}
     \end{subfigure}
     \caption{Profiles (left) and spectrum (right) of quadratic-cubic Maxwell fronts of length $2$ and $3$ for \(C = 0.02\).
     }
     \label{fig:2-3_fronts_length_23}
 \end{figure}
The numerical results for codes of length $N \leq 3$ indicate that the only Maxwell fronts which persist for all values of $C$ are the two simplest configurations of length $0$ and $1$, they are the intersite and onsite configurations respectively. Consequently, fronts of length $\geq 2$ such as illustrated in Fig. \ref{fig:2-3_fronts_length_23} have no continuum counterpart. The stability of the intersite and onsite fronts also remains unchanged as $C$ increases. Figure \ref{fig:2-3_maxwell_fronts_eigval_C=0_to_1.1.pdf} shows the numerically computed eigenvalue with the largest real part of the linearized operator $\mathcal{L}$ for $0 \leq C \leq 1.1$ for the onsite and intersite Maxwell fronts. The eigenvalue of the intersite front is of order $10^{-8}$ for all values of $C$, and thus is numerically zero. On the other hand, the onsite front is unstable for $C > 0$. The real part of the unstable eigenvalue increases (as predicted by the asymptotics \eqref{eqn:ac_lambda_approx}) until $C \approx 0.043$, after which, it decreases in an exponential-like manner in $C$, as shown in the inset of Fig. \ref{subfig:2-3_onsite_eigval_C=0_to_1.1}. As $C \geq 0.56$, the eigenvalue is smaller than $10^{-6}$ and the numerics fail to capture the eigenvalue at this precision. This is due to the limitation of memory and computational power.
\begin{figure}
    \centering
    \subfloat[Eigenvalue of quadratic-cubic onsite Maxwell front]{\includegraphics[scale=0.335]{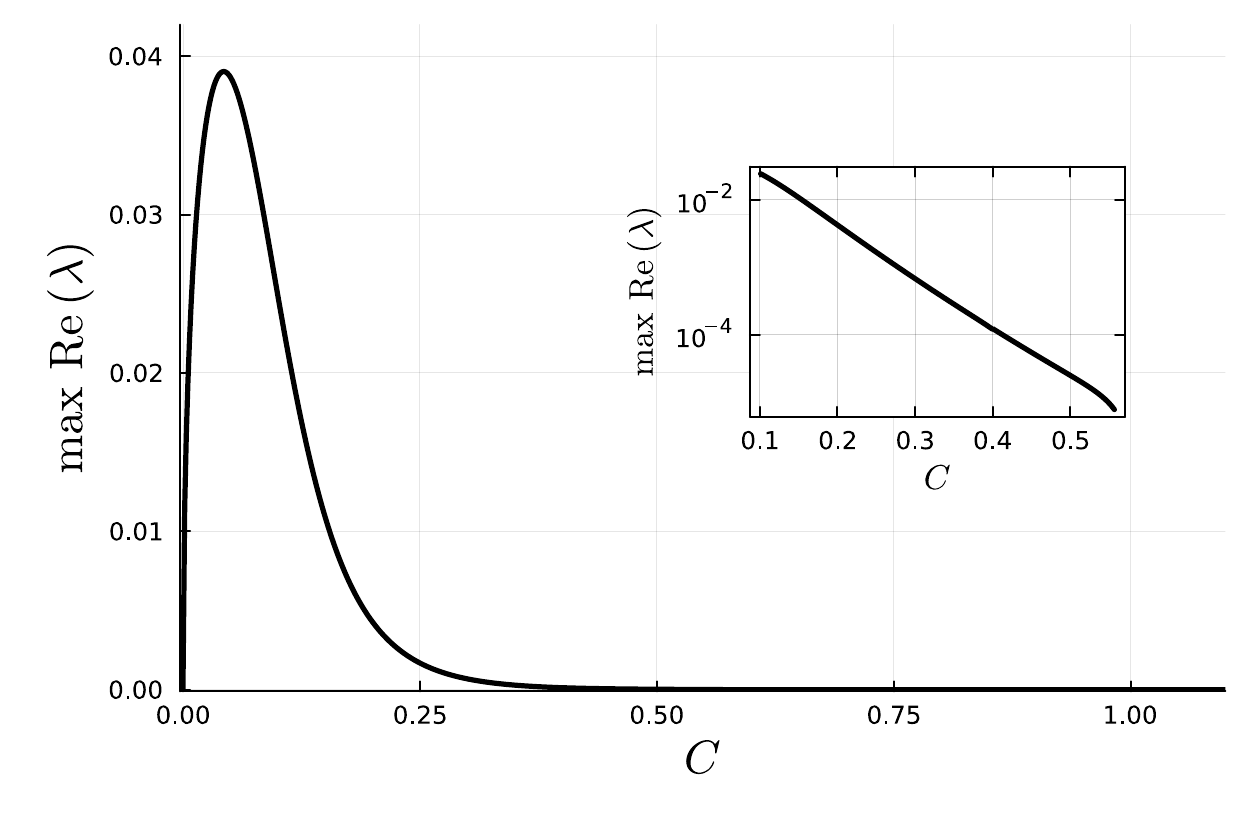}\label{subfig:2-3_onsite_eigval_C=0_to_1.1}}
    \subfloat[Eigenvalue of quadratic-cubic intersite Maxwell front]{\includegraphics[scale=0.335]{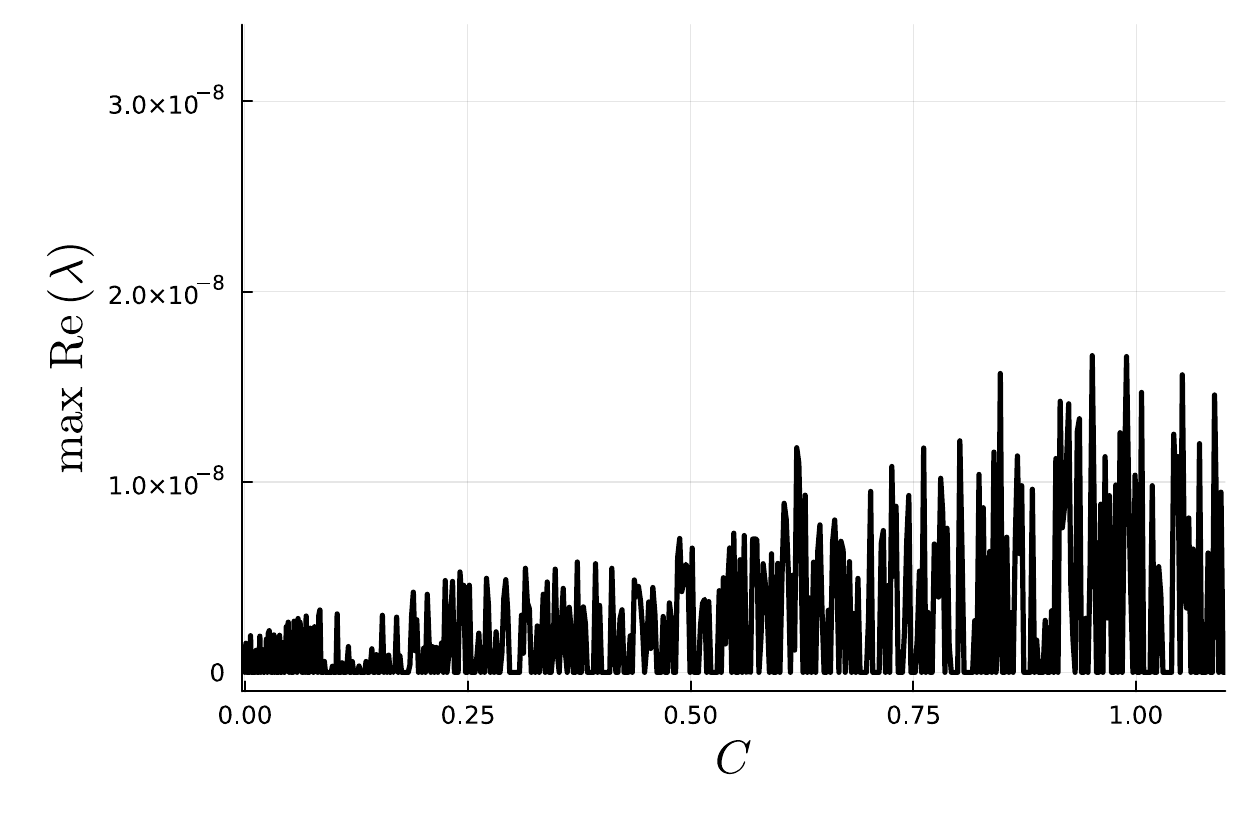}\label{subfig:2-3_intersite_eigval_C=0_to_1.1}}
    \caption{Eigenvalue of the largest real part of onsite and intersite quadratic-cubic Maxwell fronts for $C \in [0,1.1]$. The inset on the left panel shows a zoomed in view of the onsite eigenvalue for $C \in [0.4,0.56]$, the vertical axis of the inset is on a logarithmic scale.}
    \label{fig:2-3_maxwell_fronts_eigval_C=0_to_1.1.pdf}
\end{figure}
We also show the profiles of the onsite and intersite Maxwell front for $C = 0.1$ and $C = 0.5$ in Figure \ref{fig:2-3_front_prof_spec_C=0.1_and_0.5}. The instability of the onsite front is due to a single pair of nonzero real eigenvalues, which agrees with the asymptotic analysis for small and large $C$.
\begin{figure}[tbhp]
    \centering
    \begin{subfigure}{1\textwidth}
         \includegraphics[width = \textwidth]{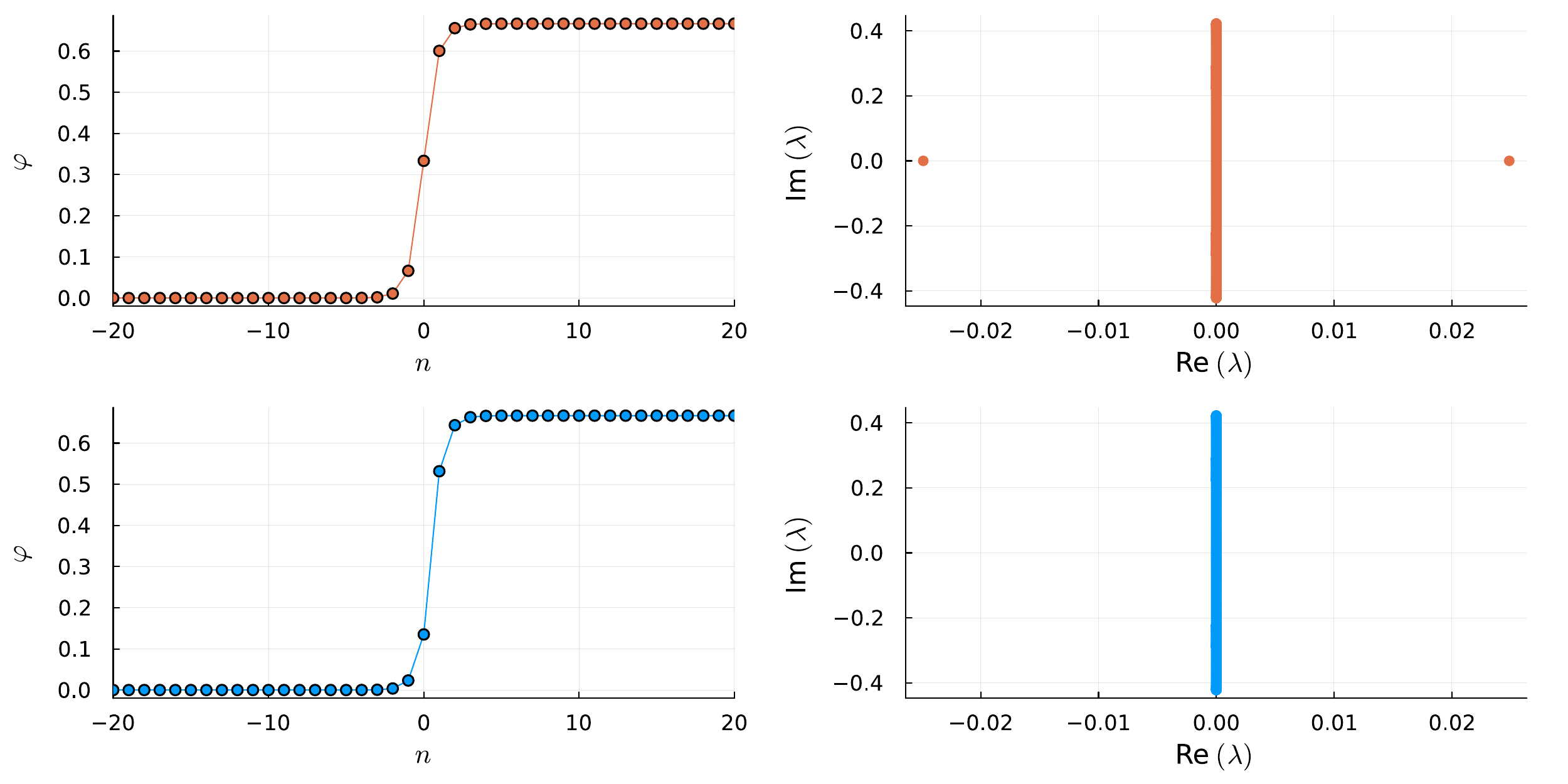}
         \caption{$C = 0.1$}
     \end{subfigure}
     \begin{subfigure}{1\textwidth}
         \includegraphics[width = \textwidth]{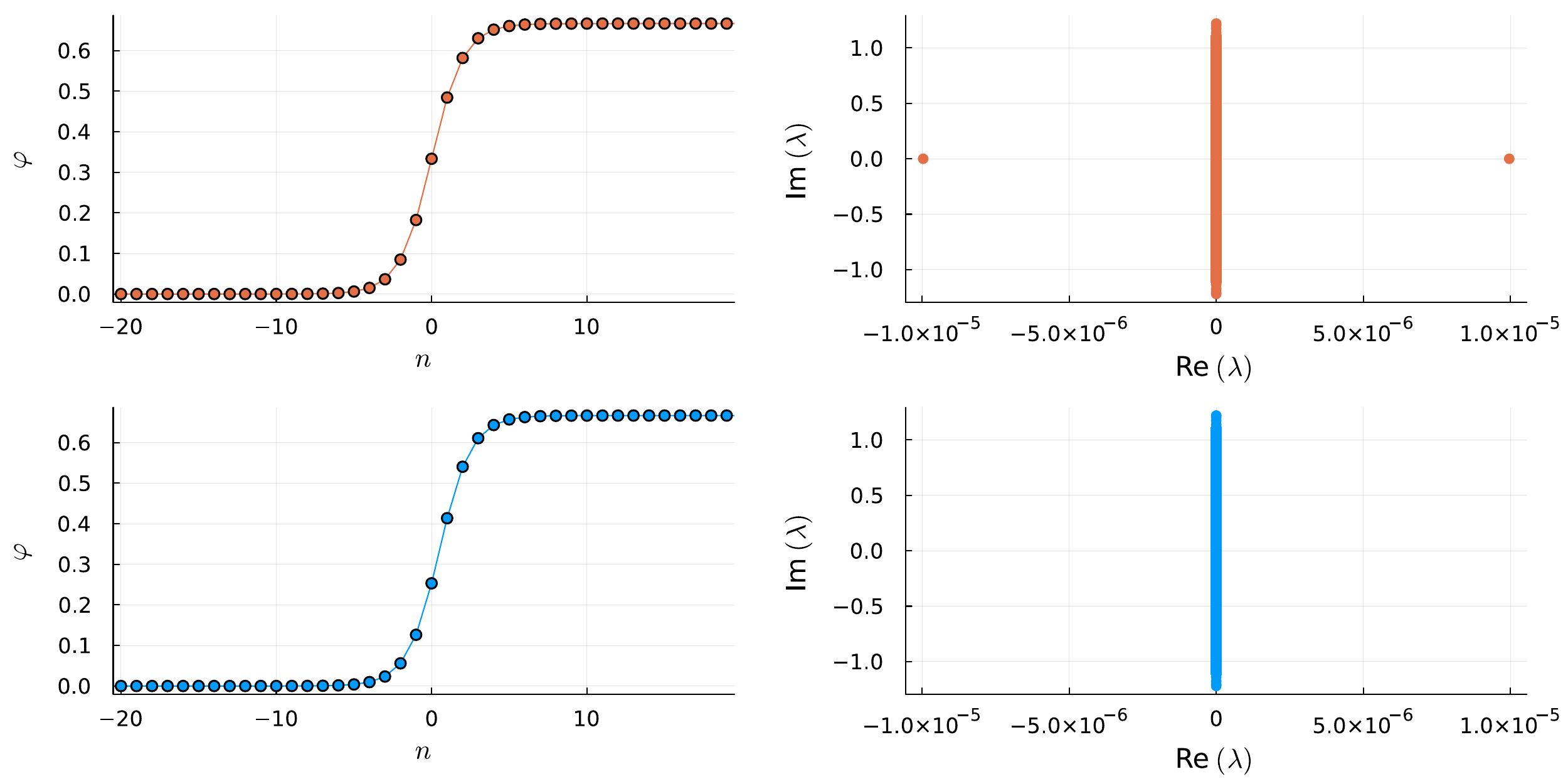}
         \caption{$C= 0.5$}
     \end{subfigure}
     \caption{Profiles (left) and spectrum (right) of the quadratic-cubic Maxwell fronts for \(C = 0.1\) and \(C = 0.5\). The orange color represents the onsite front and the blue color represents the intersite front. The presence of a spectrum with nonzero real part indicates instability of the onsite front. 
     }
    \label{fig:2-3_front_prof_spec_C=0.1_and_0.5}
\end{figure}

The time evolution of a slightly perturbed unstable onsite Maxwell front is shown in Figure \ref{fig:2-3_time_dynamics_onsite_C=0.02}. The small perturbation is chosen in the direction of the unstable eigenvector, i.e. the eigenvector corresponding to the unstable eigenvalue. Initially, the central state $\phi_{0}$ grows in magnitude. This is followed by radiation propagating towards the upper uniform state, while the lower uniform state remains mostly unaffected.

\begin{figure}[tbhp]
    \centering
    \includegraphics[width=0.6\linewidth]{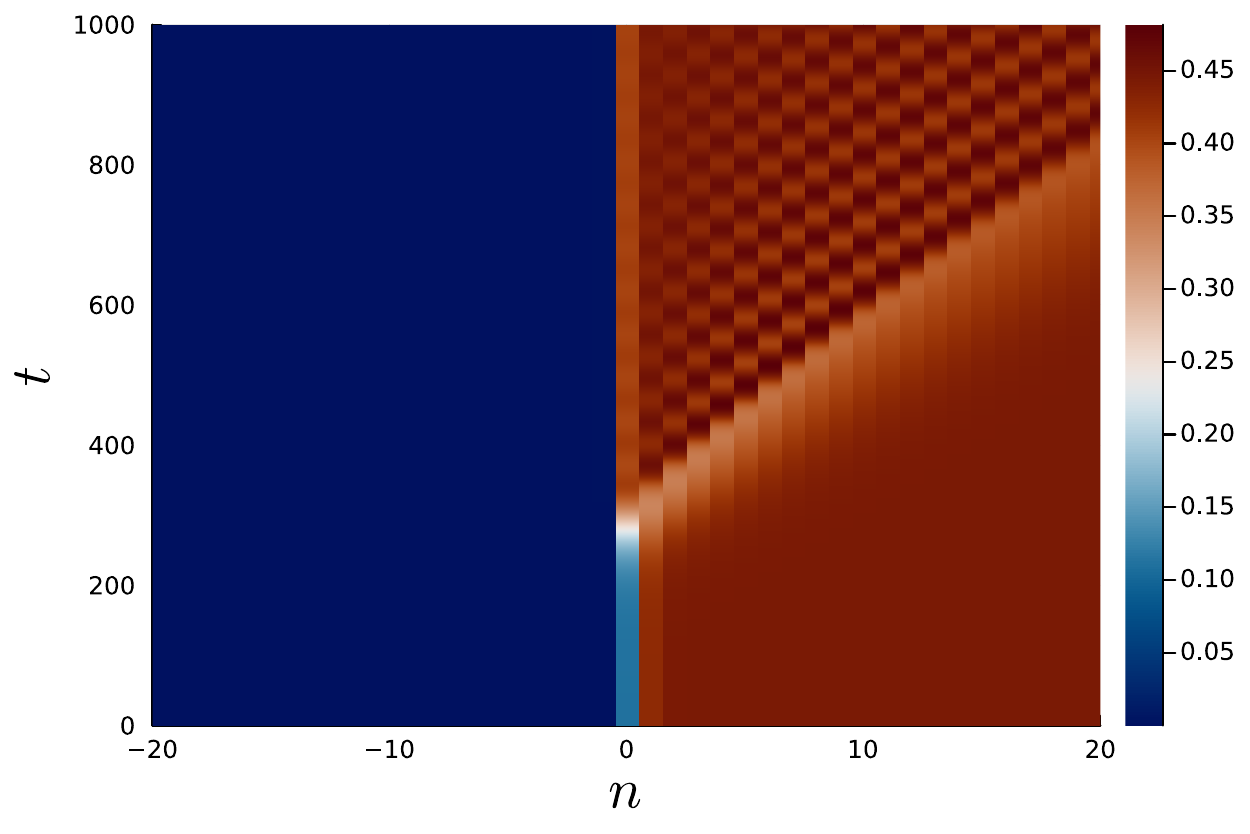}
    \caption{Time evolution of $|\phi_{n}|$ of a slightly perturbed onsite quadratic-cubic Maxwell front for $C = 0.02$.}
    \label{fig:2-3_time_dynamics_onsite_C=0.02}
\end{figure}

We also compute the bifurcation diagram of the onsite and intersite fronts existing not at the Maxwell point, $\mu \neq \mu_{M}$. These are illustrated in Figure \ref{fig:2-3_bifurcation_vary_mu}. When $\mu \neq \mu_{M}$, the branches of the onsite and intersite fronts are connected through a fold bifurcation at a certain value of $C$. As the value of $\mu$ approaches the Maxwell point, the fold bifurcation point increases and approaches infinity at the Maxwell point, at which the onsite and intersite branches are separated. Complementing analysis on the existence and stability of various discrete bright solitons (i.e., where $\varphi_n\to0$ as $n\to\pm\infty$) in competing nonlinearities has been presented recently in \cite{alfimov2025intrinsic,alfimov2025stability}.

\begin{figure}[tbhp]
    \centering
    \includegraphics[width=0.6\linewidth]{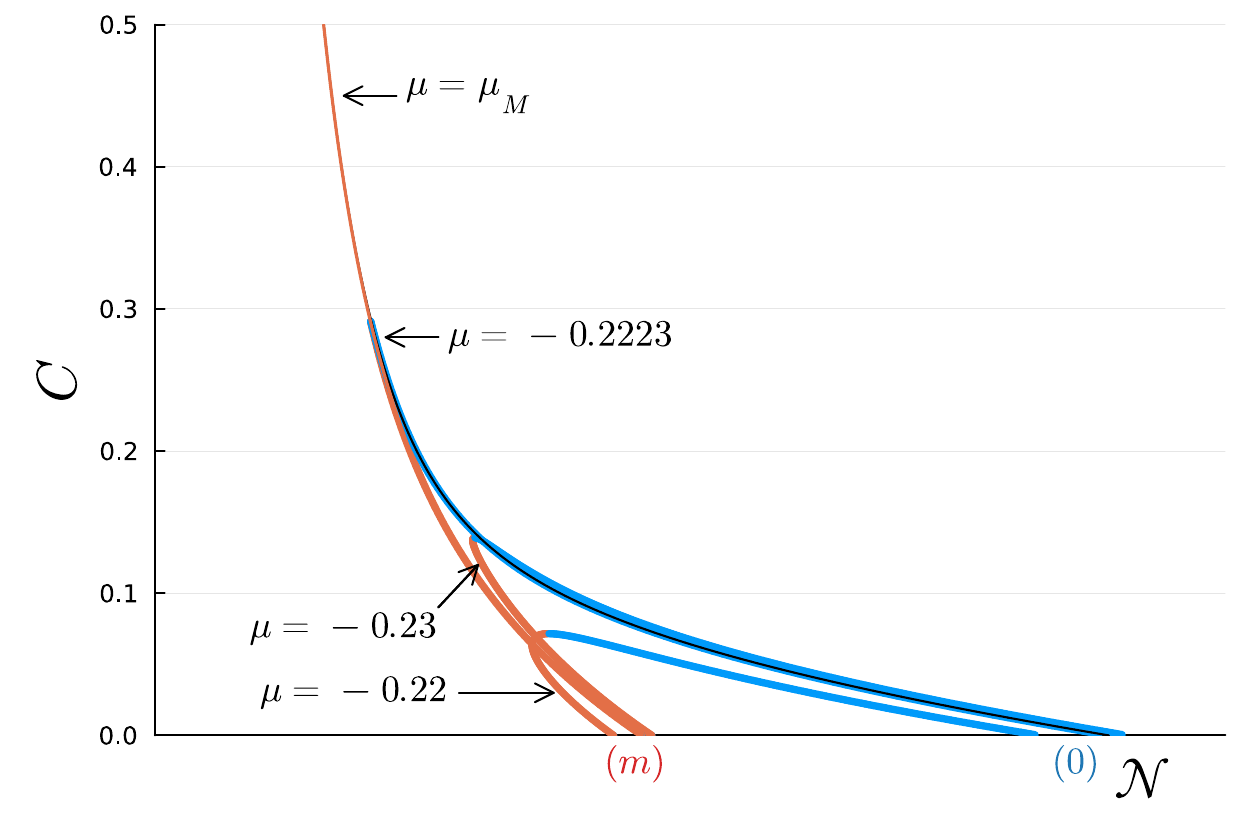}
    \caption{Bifurcation diagram of onsite and intersite fronts for various values of $\mu$. The thick red and blue lines represent unstable and stable fronts respectively. The thin black line represents the branch of the onsite Maxwell front, and the thin red line represents the intersite Maxwell front.}
    \label{fig:2-3_bifurcation_vary_mu}
\end{figure}

\section{Cubic-Quintic Nonlinearity} \label{sec4}
We now consider the cubic-quintic discrete nonlinear Schrödinger equation,
\begin{equation}\label{eqn:3-5_dnls}
    i \dot{\phi}_{n} = -\dfrac{C}{2}\Delta\phi_{n} - \mu \phi_{n} - |\phi_{n}|^{2}\phi_{n} + |\phi_{n}|^{4}\phi_{n},
\end{equation}
which also demonstrates multistability between the uniform solutions
\begin{equation}\label{eqn:3-5_uniform_solutions}
    \varphi^{\mathrm{low}} = 0,
    \quad \varphi^{\mathrm{mid}} = \left(\dfrac{1-\sqrt{1+4\mu}}{2}\right)^{1/2},
    \quad \varphi^{\mathrm{up}} = \left(\dfrac{1+\sqrt{1+4\mu}}{2}\right)^{1/2}.
\end{equation}
As in the quadratic-cubic case, the lower and upper uniform solutions, $\varphi^{\mathrm{low}}$ and $\varphi^{\mathrm{up}}$, are always stable, whereas the middle uniform solution $\varphi^{\mathrm{mid}}$ is always unstable.
The corresponding Hamiltonian is
\begin{equation}
    E(\phi) = \sum_{n=-\infty}^{\infty}\left[\dfrac{C}{2}|\phi_{n}-\phi_{n-1}|^{2} + \dfrac{\mu}{2}|\phi_{n}|^{2} + \dfrac{1}{4}|\phi_{n}|^{4} - \dfrac{1}{6}|\phi_{n}|^{6}\right].
\end{equation}
The linearized operators $L_{\pm}$ are
\begin{align}
    L_{-} &= -\dfrac{C}{2}\Delta - \mu - |\varphi_{n}|^{2} + |\varphi_{n}|^{4}, \label{eqn:3-5_Lm} \\
    L_{+} &= -\dfrac{C}{2}\Delta -\mu - 3|\varphi_{n}| + 5|\varphi_{n}|^{4}. \label{eqn:3-5_Lp}
\end{align}
In this case, the Maxwell point is $\mu_{M} = -3/16$, where $\varphi^{\mathrm{low}} = 0$ and $\varphi^{\mathrm{up}} = \sqrt{3}/2$ are energetically identical. Thus, Maxwell fronts connecting these two uniform solutions exist up to the continuum limit.
Having a similar structure to the quadratic-cubic nonlinearity, where the front connects a zero uniform solution to a positive uniform solution, a similar analysis applies to the onsite and intersite Maxwell fronts in the cubic-quintic nonlinearity case analogously. The profile and spectrum of such fronts are presented in Figure \ref{fig:3_5_front_profile_spectrum}.

\begin{figure}[tbhp]
     \centering
     \begin{subfigure}{1\textwidth}
         \includegraphics[width = \textwidth]{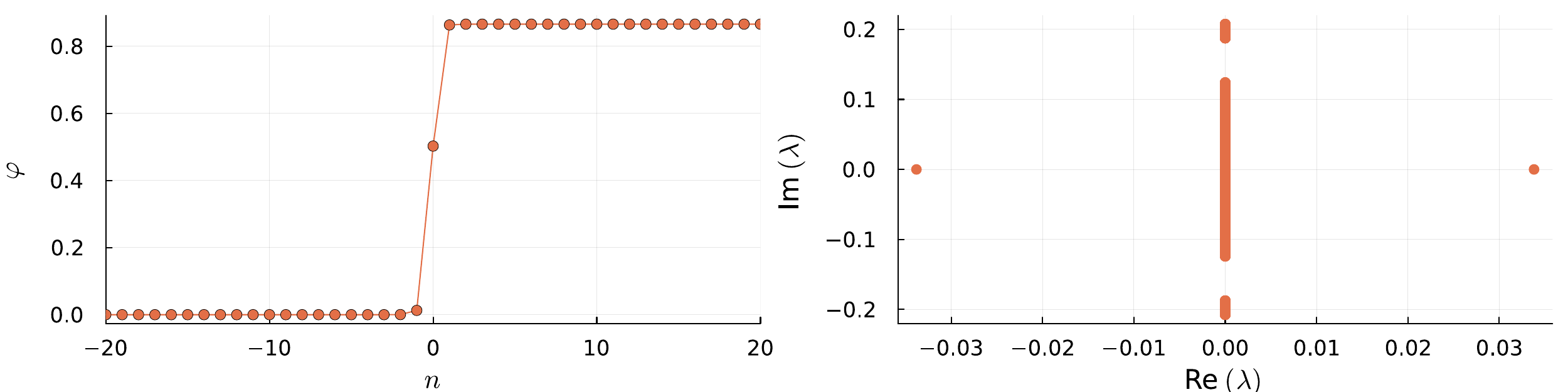}
         \caption{Onsite Maxwell front}
     \end{subfigure}
     \begin{subfigure}{1\textwidth}
         \includegraphics[width = \textwidth]{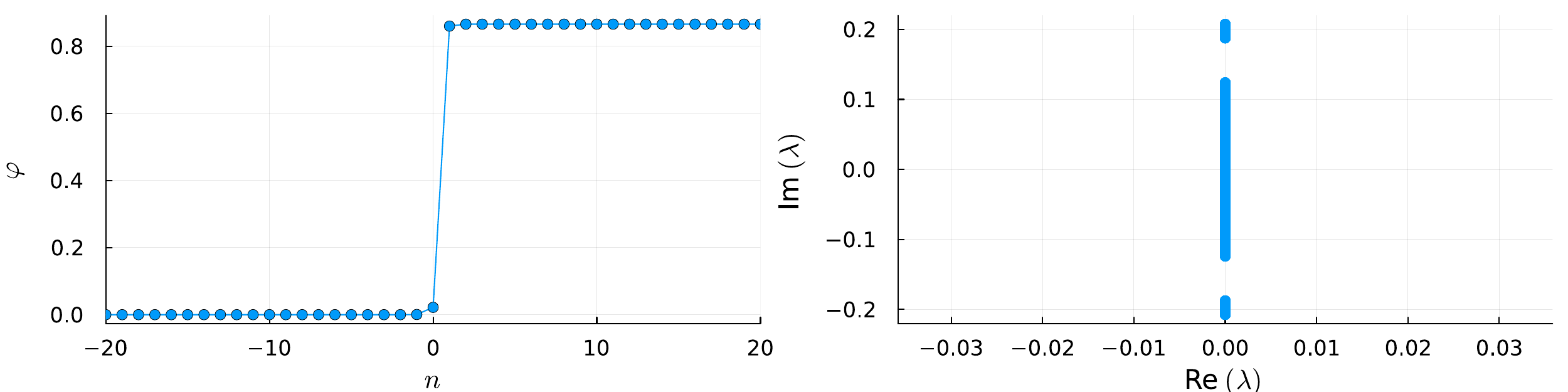}
         \caption{Intersite Maxwell front}
     \end{subfigure}
     \caption{Profiles (left) and spectrum (right) of the cubic-quintic Maxwell fronts for \(C = 0.01\). The presence of a spectrum with nonzero real part indicates instability of the onsite front. 
     }
     \label{fig:3_5_front_profile_spectrum}
 \end{figure}

In the anticontinuum limit, the spectrum of the corresponding linearized operator $\mathcal{L}$ consists of $\lambda = 0$ for the nonzero sites and $\lambda = \pm i |\mu_{M}|$ for the zero sites. In the case of onsite Maxwell fronts, the operator $L_{+}$ also possesses a negative eigenvalue at the anticontinuum limit, which persists under small perturbations in $C$, i.e. when $0 < C \ll 1$. Thus, Theorem \ref{thm:eigval_count} predicts the instability of the onsite Maxwell front due to a pair of real eigenvalues of $\mathcal{L}$. In contrast, the intersite Maxwell front is stable because the operator $L_{+}$ corresponding to the intersite configuration is a positive multiplicative operator when $C = 0$. Hence, no negative eigenvalue of $L_{+}$ exists for small $C$, guaranteeing its stability for small $C$. We may also obtain an approximation of the real eigenvalues of $\mathcal{L}$ for the onsite configuration by using a similar decaying ansatz as in the quadratic-cubic case. The approximation up to leading order in $C$ is
\begin{equation} \label{eqn:3-5_ac_lambda_approx}
    \lambda^{2} = -\frac{9  {\left(\sqrt{3} - 2 \cdot 3^{\frac{1}{4}} + 1\right)}}{8 \, {\left(\sqrt{3} {\left(3^{\frac{3}{4}} - 3\right)} - 3 \cdot 3^{\frac{3}{4}} + 9\right)}}C + O(C^{2}).
\end{equation}
Similarly, we may obtain the one-site approximation for the eigenvalues as
\begin{equation} \label{eqn:3-5 ac lambda approx one site}
    \lambda^{2} = \frac{\sqrt{3}}{8}C + O(C^{2}).
\end{equation}
These approximations are compared with the numerically calculated eigenvalues in Figure \ref{fig:3-5_ac_eigval_compare}, showing a good agreement with the numerically computed eigenvalues. As in the quadratic-cubic case, both the one-site approximation \eqref{eqn:3-5 ac lambda approx one site} and the asymptotic approximation from the decaying ansatz \eqref{eqn:3-5_ac_lambda_approx} agree with the numerics asymptotically as $C \to 0$. However, the approximation from the decaying ansatz \eqref{eqn:3-5_ac_lambda_approx} is more accurate for slightly larger values of $C$.

\begin{figure}[htbp]
    \centering
    \includegraphics[width=0.6\linewidth]{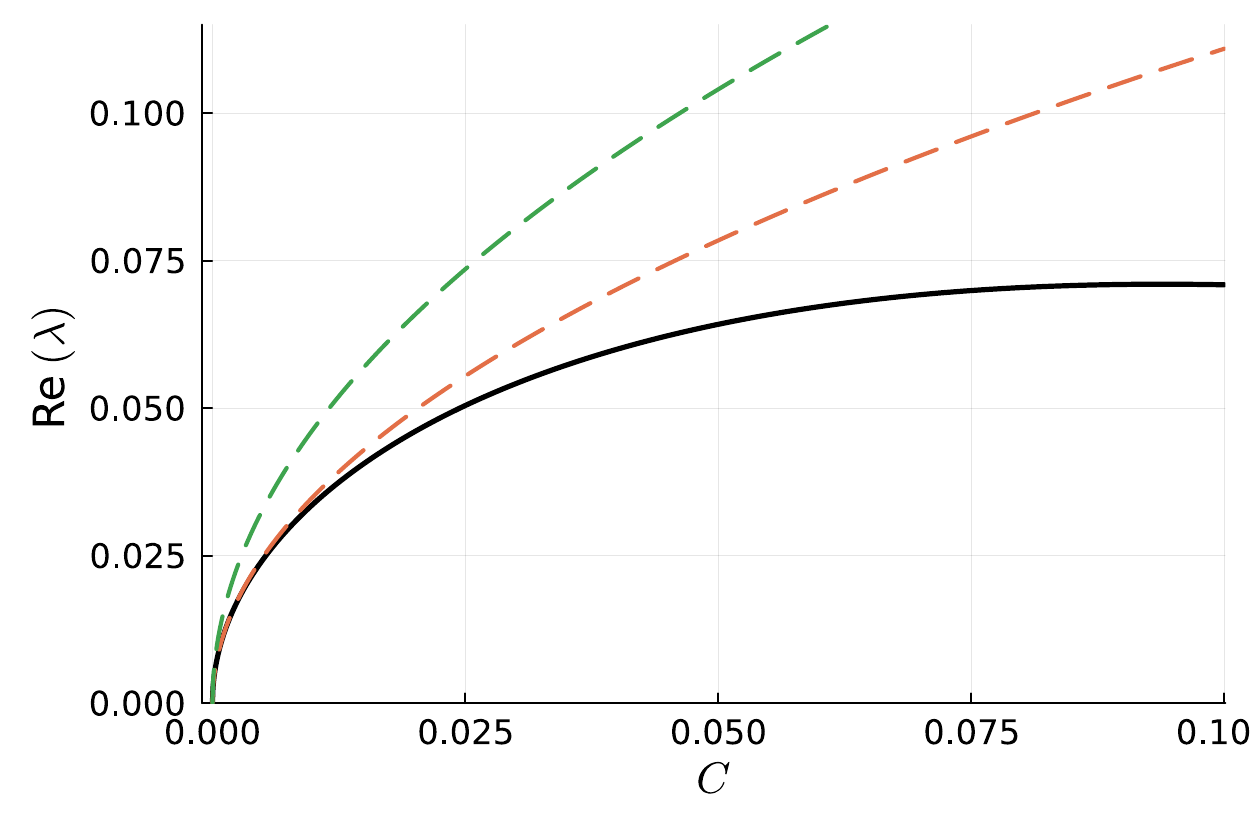}
    \caption{Real part of the eigenvalue of $\mathcal{L}$ for the onsite cubic-quintic Maxwell front for $C \in [0,0.1]$. The black solid line shows the numerical result, the orange dashed line shows the asymptotic approximation \eqref{eqn:3-5_ac_lambda_approx}, and the green dashed line shows the one-site approximation \eqref{eqn:3-5 ac lambda approx one site}.}
    \label{fig:3-5_ac_eigval_compare}
\end{figure}
\par
Near the continuum limit $(C \gg 1)$, a similar exponential asymptotic analysis procedure leads to the asymptotic solution of the stationary Maxwell front as
\begin{equation} \label{eqn:3-5_varphi_expansion_truncated_with_remainder}
        \varphi(x) \sim \varphi_{0}(x) + \sum_{j=1}^{N-1} h^{2j} \varphi_{j}(x) + R_{N}(x), \ \ (h \ll 1, x > 0).
\end{equation}
where $h = \sqrt{2/C}$ as before. The leading-order solution, $\varphi_{0}(\tilde{x})$, written in the shifted spatial variable $\tilde{x} = x - x_{0}$, is the continuum heteroclinic orbit:
\begin{equation} \label{eqn:3-5_cont_varphi0_leading_order}
        \varphi_{0}(\tilde{x}) = \frac{\sqrt{3}}{2}\left(1 + e^{-\sqrt{3} \tilde{x}/2}\right)^{-1/2}.
\end{equation}
The poles of $\varphi_{0}(\tilde{x})$ are $\zeta = \zeta_{k} = i 2\pi(1 + 2k)/\sqrt{3}$, $k \in \mathbb{Z}$, where $\varphi_{0}(\tilde{x}) \sim O(\tilde{x}-\zeta)^{-1/2}$ near the poles. The remainder term $R_{N}(\tilde{x})$ in this case is given by
\begin{equation}
    R_{N}(\tilde{x}) \sim -2\pi \Lambda_{1} h^{-3} e^{-\frac{4\pi^{2}}{\sqrt{3} h}} \left[ \text{Re}(G(\tilde{x})) \sin\left(\frac{2\pi \tilde{x}}{h}\right) - \text{Im}(G(\tilde{x})) \cos\left(\frac{2\pi \tilde{x}}{h}\right) \right], \quad (h \ll 1).
\end{equation}
where $G(x) = \varphi_{0}'(x)\int_{\zeta}^{x} [\varphi_{0}'(s)]^{-2} ds$ and the Stokes constant $\Lambda_{1}$ is $\Lambda_{1} \approx -89$. The important aspect of the asymptotic solution \eqref{eqn:3-5_varphi_expansion_truncated_with_remainder} is its behavior in the far-field $\tilde{x} \to \infty$, which in this case is dominated by the remainder term:
\begin{equation}
    \varphi(\tilde{x}) \sim \dfrac{16\pi}{3\sqrt{3}} \Lambda_{1} h^{-3} e^{-\frac{4\pi^{2}}{\sqrt{3} h}} e^{\sqrt{3} \tilde{x} / 2} \sin\left(\frac{2\pi \tilde{x}}{h}\right), \quad (h \ll 1, \tilde{x} \gg 1),
\end{equation}
from which we see that $\varphi(\tilde{x})$ remains bounded if $n = 0$ or $n = 1/2$ (modulo integer shifts), corresponding to the onsite and intersite Maxwell fronts, respectively.
\par
Analysis of the bifurcation of the zero eigenvalue of $L_{+}$ near the continuum limit can be done in a similar manner. It is found that, similar to the quadratic-cubic case, the linearized operator $L_{+}$ has an exponentially small negative eigenvalue for onsite Maxwell fronts, and a positive one for intersite Maxwell fronts. Therefore, the stability of the onsite and intersite Maxwell fronts is the same as in the near-anticontinuum limit $(0 < C \ll 1)$, i.e., the onsite is unstable and the intersite is stable.

As in the quadratic-cubic case, Maxwell fronts with codes of length $> 1$ do not survive as $C \to \infty$ as they undergo fold bifurcation and is connected to branches of other fronts of length $> 1$. The bifurcation diagram for cubic-quintic Maxwell fronts of code lengths $N \leq 2$ is presented in Figure \ref{fig:3-5_bifurcation_diag_maxwell_012}. However, the qualitative behavior of the bifurcation is not necessarily the same as in the quadratic-cubic case. For example, in the cubic-quintic case, the $(mm)$ branch is connected to the $(ml)$ branch, whereas it was connected to the $(ul)$ branch in the quadratic-cubic case. In the cubic-quintic case, it is the $(um)$ branch which is connected to the $(ul)$ branch.

\begin{figure}[tbhp]
    \centering
    \includegraphics[width=0.6\linewidth]{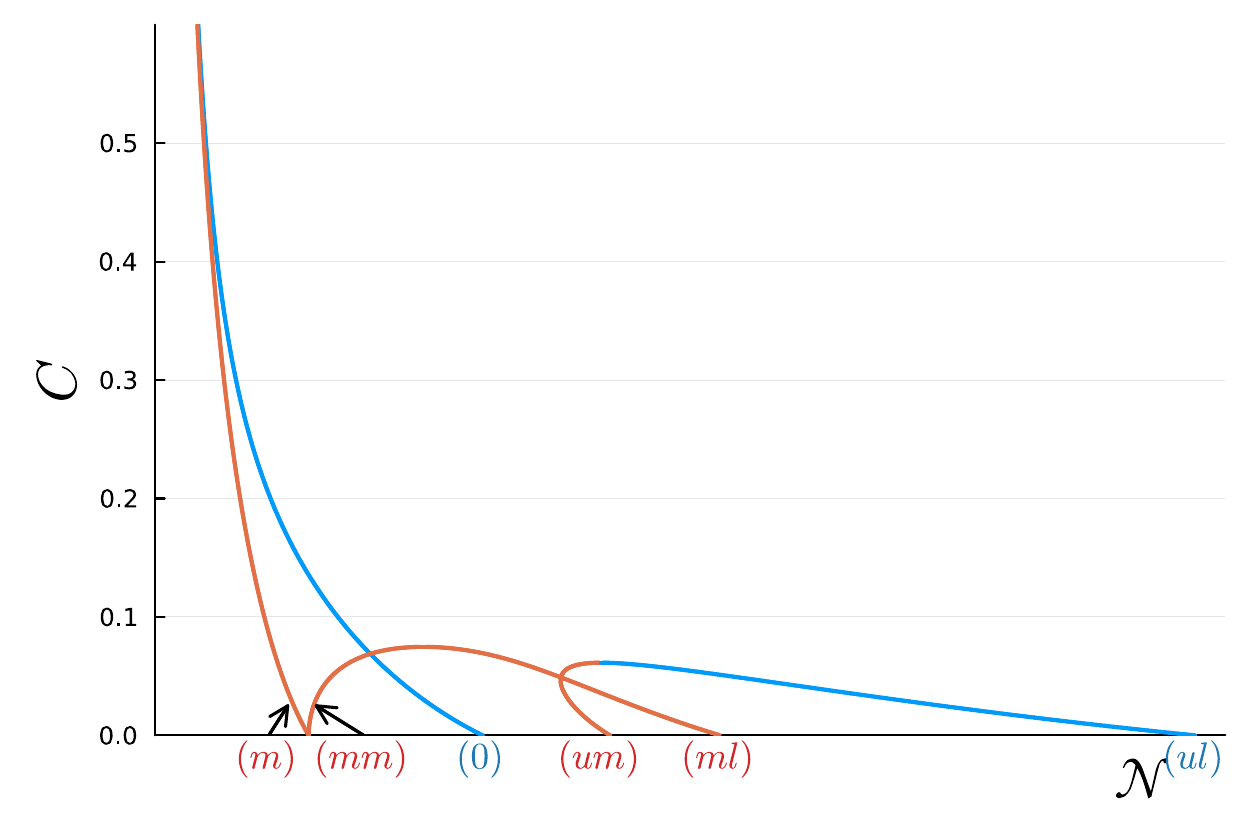}
    \caption{Bifurcation diagram of cubic-quintic Maxwell fronts with codes of length $0$, $1$, and $2$. The blue lines represent stable fronts and the red lines represent unstable fronts.}
    \label{fig:3-5_bifurcation_diag_maxwell_012}
\end{figure}

We note the differences between the cubic-quintic Maxwell fronts and cubic-quintic fronts of the dark soliton type, which were studied in \cite{Maluckov2007}. In the latter case, it was shown that the onsite cubic-quintic dark soliton can either be marginally stable, oscillatory unstable, or exponentially unstable (instability due to nonzero real eigenvalues) depending on the value of the parameter $\mu$. The intersite cubic-quintic dark-soliton, on the other hand, is always unstable. These differ from the case of Maxwell fronts: the onsite cubic-quintic Maxwell front is unstable due to a pair of real eigenvalues, whereas the intersite cubic-quintic Maxwell front is always stable, as established by the preceding analysis.

\section{Conclusion}\label{sec5}

In this work, we have investigated the stability and bifurcation behaviors of Maxwell fronts in DNLS equations with quadratic-cubic and cubic-quintic nonlinearities. By analyzing the eigenvalue spectra of the corresponding linearized operators, we have gained key insights into the dynamics of these front solutions across various coupling regimes.

For both the quadratic-cubic and cubic-quintic cases, we observed that uniform solutions exhibit multistability and can transition into Maxwell fronts as the coupling becomes nonzero. The stability analysis revealed that the onsite configuration is unstable, while the intersite configuration remains stable, in contrast to the well-established stability properties of discrete dark solitons.

While the current study focuses on the one-dimensional DNLS equation with specific nonlinearities, such as quadratic-cubic and cubic-quintic, an important direction for future work is to extend these investigations to higher-dimensional lattices. In higher dimensions, Maxwell fronts could exhibit more complex dynamical behaviors. Understanding their stability, bifurcation, and multistability in two or more dimensions could offer valuable insights into more intricate physical systems. Additionally, exploring nonlocal interactions, where coupling extends beyond nearest neighbors, may uncover new phenomena. These extensions are relevant to fields such as nonlinear photonics, where multi-dimensional and nonlocal effects are commonly observed. Further studies on these topics will be reported in future work.

\section*{acknowledgements}

The authors would like to express their sincere gratitude to Dr.\ Rudy Kusdiantara for his invaluable discussions and insightful suggestions during the initial stages of this research. HS acknowledges support by Khalifa University through Research \& Innovation Grants (No.\ 8474000617/RIG-S-2023-031 and No.\ 8474000789/RIG-S-2024-070). During the preparation of this work, the authors used Grammarly and ChatGPT to enhance language and readability. After using these tools/services, the authors reviewed and edited the content as needed and take full responsibility for the content of the publication.


\section*{conflict of interest}

The authors declare that they have no conflict of interest regarding the publication of this work.


\printendnotes

\bibliography{sample}

\end{document}